\documentclass{article}

\usepackage{vmargin,epsfig}
\usepackage[english]{babel}
\usepackage[latin1]{inputenc}
\usepackage[tbtags]{amsmath}
\usepackage{amsthm,amssymb}
\usepackage[all]{xy}
\usepackage{times}
\usepackage{microtype}
\usepackage{url}
\usepackage[algoruled,vlined,onelanguage,linesnumbered]{algorithm2e}

\newtheorem{theo}{Theorem}[section]
\newtheorem{lemma}[theo]{Lemma}
\newtheorem{prop}[theo]{Proposition}
\newtheorem{cor}[theo]{Corollary}
\theoremstyle{definition}
\newtheorem{deftn}[theo]{Definition}

\newtheorem{problem}[theo]{Problem}
\theoremstyle{remark}
\newtheorem{rem}[theo]{Remark}

\def\leq{\leqslant}
\def\geq{\geqslant}

\def\A{\mathbb{A}}
\def\N{\mathbb{N}}
\def\Z{\mathbb{Z}}

\def\E{\mathbb{E}}
\def\calF{\mathcal{F}}
\def\calM{\mathcal{M}}
\def\Prob{\mathbb{P}}

\def\O{\mathcal{O}}

\def\epsilon{\varepsilon}

\def\pr{\text{\rm pr}}
\def\dist{\text{\rm dist}}
\def\card{\text{\rm Card}}
\def\GL{\text{\rm GL}}
\def\Diag{\text{\rm Diag}}
\def\t{{}^{\text{t}}}

\def\partd{\underline{d}}

\title{Random matrix over a DVR and LU factorization}
\author{Xavier Caruso}
\date\today

\begin{document}

\maketitle

\begin{abstract}
Let $R$ be a discrete valuation ring (DVR) and $K$ be its fraction 
field. If $M$ is a matrix over $R$ admitting a LU decomposition, it 
could happen that the entries of the factors $L$ and $U$ do not lie in 
$R$, but just in $K$. Having a good control on the valuations of these 
entries is very important for algorithmic applications. In the paper, we 
prove that in average these valuations are not too large and explain how 
one can apply this result to provide an efficient algorithm computing a 
basis of a coherent sheaf over $\A^1_K$ from the knowledge of its 
stalks.
\end{abstract}

\setcounter{tocdepth}{2}
\tableofcontents

\bigskip

\noindent
\hrulefill

\bigskip

Throughout the paper, we fix a ring $R$ equipped with a discrete 
valuation $v_R : R \to \N \cup \{\infty\}$. We assume that $v_R$ is 
normalized so that it takes the value $1$ and we fix an element $\pi 
\in R$ such that $v_R(\pi) = 1$. We also assume that $R$ is complete 
with respect to the distance defined by $v_R$. The residue field of $R$ 
and its fraction field are denoted by $k$ and $K$ respectively. The 
valuation $v_R$ extends uniquely to $K$ and we continue to denote this 
extension by $v_R$. We finally set $q = \card\:k$ and assume that 
$q$ is finite. Two typical examples of this are (1) $R = \Z_p$ (the ring 
of $p$-adic integers) equipped with the usual $p$-adic valuation and (2) 
$R = k[[x]]$ (the ring of power series) where $k$ is a finite 
field.

If $d$ is a positive integer, we denote by $\Omega$ the ring of square 
matrices of size $d$ with coefficients in $R$. It is a compact additive 
group whose Haar measure is denoted by $\mu$. We assume that $\mu$ is 
normalized so that $(\Omega, \mu)$ is a probability space (\emph{i.e.} 
$\mu(\Omega) = 1$). Thus, it makes sense to study some statistics on 
$\Omega$. Surprinsingly, the literature around this subject seems to be 
very poor. Nevertheless related questions were already addressed by 
Abdel-Ghaffar in \cite{abdel} and Evans in \cite{evans}: the main result 
of \cite{abdel} is the computation of the law of the random variable 
``valuation of the determinant'' (in the case where $R$ is a power 
series ring but his argument works for a more general discrete valuation 
ring) whereas, in \cite{evans}, Evans studies the random variable 
``valuation of the elementary divisors''.

In this paper, we are mainly interested in the random variable $V_L$: 
``valuation of the $L$-part in the LU decomposition''. We give several 
estimations of its law, its expected value and its standard deviation. 
Roughly speaking, we prove that $\E[V_L] = \log_q d + O(1)$ and 
$\sigma(V_L) = O(1)$ where the notation $O(1)$ refers to a quantity 
bounded by a universal constant. We also bound from above the 
probability that $V_L$ deviates from its expectation. For more precise 
statements, we refer to Theorem \ref{th:VLexp}, Theorem \ref{th:VLlaw} 
and Corollary \ref{cor:VLdev} in the introduction of \S \ref{sec:stats}.

In \S \ref{sec:algo}, we move to algorithmic applications. Firstable, we 
propose in \S \ref{subsec:stableLU} a \emph{stable} algorithm to compute 
a LU decomposition of a matrix over $R$ (unfortunately standard Gauss 
elimination is far for being stable) and analyze closely the losses of 
precision it generates in average (which turn out to be optimal in some 
sense). \S \ref{subsec:simulPLU} is devoted to a study of the notion of 
``simultaneous PLU decomposition'', which will play an important role 
for our next (and main) application presented in \S \ref{subsec:sheaf}. 
This application is of geometric nature. We let $X = \A^1_K$ be the 
affine line over $K$. Recall that a coherent subsheaf of $\calF \subset 
\O_X^d$ (where $d$ is some integer) is determined by the data of all its 
stalks $\calF_x \subset \O_{X,x}^d$ at all closed points $x \in X$. 
Furthermore, we know that all such sheaves $\calF$ as above admit a 
global basis. In \S \ref{subsec:sheaf}, we describe an algorithm that 
computes a basis of $\calF$ knowing all its stalks and, once again, 
analyze its stability (which will turn out to be rather good).

\section{Some statistics related to LU decomposition}
\label{sec:stats}

If $A$ is a (commutative) ring, we shall denote by $M_d(A)$ the ring of 
square $d \times d$ matrices. Recall that we have endowed $\Omega = 
M_d(R)$ with its Haar measure. Choosing a matrix at random with respect 
to this measure is just choosing independently each entry at random with 
respect to the Haar measure on $R$. Furthermore, since $R$ is complete, 
every element $x \in R$ can be written uniquely as an infinite sum $x = 
\sum_{i=0}^\infty a_i \pi^i$ where the coefficients $a_i$'s are taken in 
a fixed set $\mathcal R \subset R$ of representatives of elements of 
$k$ (\emph{i.e.} the restriction to $\mathcal R$ of the canonical 
projection $R \to k$ is bijective) and conversely, any sum 
$\sum_{i=0}^\infty a_i \pi^i$ as above converges and then defines an 
element in $R$. With this description, generating a random element (with 
respect to the Haar measure) of $R$ is just choosing at random all 
coefficients $a_i$'s in $\mathcal R$ independently and uniformly.

We shall say that a matrix $M \in M_d(K)$ admits a \emph{LU 
decomposition} if it can be factorized as a product $L(M)\cdot U(M)$ 
where:
\begin{itemize}
\item $L(M)$ is a unit\footnote{It means that all diagonal 
entries are equal to $1$.} lower triangular matrix with coefficients
in  $K$, and
\item $U(M)$ is a upper triangular matrix with coefficients in $K$.
\end{itemize}
We underline that, even if $M$ has coefficients in $R$, we do not 
require that $L(M)$ and $U(M)$ belong to $M_d(R)$. Here are some well 
known facts: (1) an invertible matrix $M$ admits a LU decomposition if 
and only if all its principal minors do not vanish and (2) when it 
exists, a LU decomposition is unique (\emph{i.e.} the matrices $L(M)$ 
and $U(M)$ are uniquely determined).
We will consider $L$ and $U$ as two partially defined functions on 
$M_d(K)$. For $\omega \in \Omega$ such that $L(\omega)$ is defined, let 
us denote by $V_L(\omega)$ the opposite of the smallest valuation of an 
entry of $L(\omega)$. The aim of this section is to study the random 
variable $V_L$. Here are the main results we will prove.

\begin{theo}
\label{th:VLexp}
Setting
\begin{equation}
\label{eq:Eqd}
E(q,d) = \sum_{v=1}^\infty \big[1 - (1-q^{-v})^d\big]
\end{equation}
we have $E(q,d) - \frac 1{q-1} < \E[V_L] \leq E(q,d)$.

Furthermore, the distance between $E(q,d)$ and $\log_q d$ is bounded by 
$\frac 1{\log(2)}$ (and by $1$ if $q \geq 3$).
\end{theo}

\begin{theo}
\label{th:VLlaw}
For all positive real number $\ell$, we have:
$$\Prob\big[ {\textstyle |V_L - \log_q d - \frac 1 2| > 
\ell + \frac 1 2} \big] \leq 
\frac q{q-1} \cdot q^{-\ell} \cdot \Big( 2 + \ell \cdot \log 
q\Big).$$
\end{theo}

\begin{cor}
\label{cor:VLdev}
The standard deviation of $V_L$ is bounded by an explicit universal 
constant (which can be chosen equal to $6.5$).
\end{cor}

\subsection{Some useful tools}

This subsection gathers some preliminaries to the proof of Theorem 
\ref{th:VLexp}, Theorem \ref{th:VLlaw} and Corollary \ref{cor:VLdev}. We 
first recall some basic facts about LU decomposition, then introduce the 
random variables $V_{i,j}$'s (which will play a crucial role in the 
sequel) and finally prove several important properties of them.

\subsubsection{Cramer's rule for LU decomposition}

Let $M \in M_d(K)$. A useful formula for our purpose is an analogue of 
Cramer's rule which gives a closed expression of the entries of $L(M)$ 
as a quotient of two determinants. This formula appears for instance in 
\cite{householder}, \S 1.4; let us recall it briefly. If $I$ and $J$ are 
two subsets of $\{1, \ldots, d\}$, we denote by $M_{I,J}$ the submatrix 
of $M$ obtained by deleting all columns and rows whose index are not in 
$I$ and $J$ respectively. The $i$-th principal minor is then the 
determinant of the matrix $M_{I,I}$ where $I = \{1, \ldots, i\}$; we 
will denote it by $\delta_i(M)$. With these notations, we have:
\begin{eqnarray}
\text{if } i > j & : &
L(M)_{i,j} = \frac{\det M_{I,J}}{\delta_j(M)}
\quad \text{where } I = \{1, \ldots, j-1, i\} \text{ and }
J = \{1, \ldots, j\} \label{eq:CramerL} \\
\text{if } i \leq j & : &
U(M)_{i,j} = \frac{\det M_{I,J}}{\delta_{i-1}(M)}
\quad \text{where } I = \{1, \ldots, i\} \text{ and }
J = \{1, \ldots, i-1, j\}. \label{eq:CramerU}
\end{eqnarray}
The proof of these formulas is not difficult. For Formula 
\eqref{eq:CramerL}, note that $L(M)_{I,J} \cdot U(M)_{J,J} = M_{I,J}$ 
provided that $J$ has the particular shape $J = \{1, \ldots, j\}$; then, 
passing to the determinant, we get $\det L(M)_{I,J} \cdot \det 
U(M)_{J,J} = \det M_{I,J}$ and the desired relation follows by combining 
these equalities for $I = J$ and $I = \{1, \ldots, j-1, i\}$. The proof 
of Formula \eqref{eq:CramerU} is similar.


\subsubsection{The random variables $V_{i,j}$}
\label{subsec:Vij}

The aim of this paragraph is to define a collection of mututally 
independent random variables $V_{i,j} : \Omega \to \N \cup \{\infty\}$ 
($1 \leq i \leq j \leq d$); they will be very useful in the sequel to 
study $V_L$.
We first construct a collection of random variables $X_{i,j} : 
\Omega \to R$ ($1 \leq i \leq j \leq d$). The construction goes by 
induction on $j$. We start with a matrix $\omega$ in $\Omega$. We first
define $X_{1,1}(\omega)$ to be the top left entry of $\omega$. We then 
enter in the second round (\emph{i.e.} $j = 2$). As before, we begin by 
letting $X_{1,2}(\omega)$ denote the $(1,2)$-th entry of $\omega$ but, 
before defining $W_{2,2}(\omega)$ we do the two following modifications 
on the matrix $\omega$:
\begin{itemize}
\item if the valuation of $X_{1,2}(\omega)$ is less than the valuation
of $X_{1,1}(\omega)$, we swap the two first columns of $\omega$ and, 
then
\item we clear the $(1,2)$-th entry of $\omega$ by adding to its
second column a suitable multiple of its first column (note that
it is always possible because if the top left entry --- which serves
as pivot --- vanishes, so does the $(1,2)$-th entry).
\end{itemize}
Doing these operations, the coefficient of $\Omega$ in position 
$(2,2)$ may have changed and we define $X_{2,2}(\omega)$ to be the
\emph{new} $(2,2)$-th entry of $\omega$.
The general induction step works along the same ideas. Assume that,
after the $(j-1)$-th step, we have ended up with a matrix $\omega$ such 
that $\omega_{i',j'} = 0$ when $i' < j' < j$. We define $X_{i,j}
(\omega)$ by induction on $i$ by applying the following process 
successively for $i = 1, 2, \ldots, j-1$:
\begin{itemize}
\item first, we set $X_{i,j}(\omega)$ to the $(i,j)$-th entry of (the 
current) $\omega$;
\item second, if the valuation of $X_{i,j}(\omega)$ is less than the 
valuation of the $(i,i)$-th entry of (the current) $\omega$, we swap the 
first row of $\omega$ with its $i$-th one;
\item third, we clear the $(i,j)$-th entry of $\omega$ by adding to
its $j$-th column a suitable multiple of its first column.
\end{itemize}
We finally let $X_{j,j}(\omega)$ denote the $j$-th diagonal entry of (the
current) $\omega$. For all $(i,j)$ with $1 \leq i \leq j \leq d$, we 
also set $V_{i,j} = v_R(X_{i,j})$ and $V_i = V_{i,i}$. The $V_{i,j}$'s 
take values in $\N \cup \{\infty\}$ and they are finite almost
everywhere. 
Algorithm \ref{algo:Vij} summarizes the construction of the $V_{i,j}$'s.

\begin{algorithm}
  \SetKwInOut{Nota}{\it Notation:}
  \SetKwInOut{Void}{}

  \Nota{$\star\,$ $\omega_{i,j}$ denotes the $(i,j)$-th entry of $\omega$}
  \Void{$\star\,$ $\omega_j$ denotes the $j$-th row of $\omega$}

  \BlankLine

  \For {$j$ from $1$ to $d$}{
    \For {$i$ from $1$ to $j-1$}{
      $V_{i,j}$ $\leftarrow$ $v_R(\omega_{i,j})$\;
      \lIf {$v_R(\omega_{i,j}) < v_R(\omega_{i,i})$}
        {swap $\omega_j$ and $\omega_i$\;}
      \lIf {$\omega_{i,i} \neq 0$}
        {$\omega_j$ $\leftarrow$ 
        $\omega_j - \frac{\omega_{i,j}}{\omega_{i,i}} \cdot \omega_i$\;}
    }
    $V_{j,j}$ $\leftarrow$ $v_R(\omega_{j,j})$\;
  }
\caption{The construction of the random variables $V_{i,j}$'s\label{algo:Vij}}
\end{algorithm}

\begin{prop}
\label{prop:Xij}
The random variables $X_{i,j}$ ($1 \leq i \leq j \leq d$) are
uniformely distributed and mutually independent.
\end{prop}

\begin{proof}
Set $I = \{ \, (i,j) \, | \, 1 \leq i \leq j \leq d \, \}$. Suppose we 
are given a family $x = (x_{i,j})_{(i,j) \in I}$ of elements of $R$.
We consider the following set:
$$\Omega(x) = \big\{ \,\, \omega \in \Omega \,\, | \,\,
X_{i,j}(\omega) = x_{i,j}, \, \forall (i,j) \in I \,\, \big\}.$$
Set $v_{i,j} = v_R(x_{i,j})$ and for all $i$, let $j_i$ denote the
first index in $\{1, \ldots, d+1-i\}$ such that $v_{i,i-1+j_i}$ is 
equal to $\min(v_{i,i}, v_{i,i+1}, \ldots, v_{i,d})$. This sequence 
of integers $(j_i)$ is the code of a certain permutation $\sigma$ of 
$\{1, \ldots, d\}$ defined by the following rule. We write all the 
integers between $1$ and $d$. We define $\sigma(1)$ to be the 
$j_1$-th written integer (that is $j_1$) and we erase it. We then
define $\sigma(2)$ to be the $j_2$-th integer which remains written
(that is $j_2$ is $j_2 < j_1$ and $j_2+1$ otherwise), we erase it 
and we continue. Let $I_x$ denote the subset of $\{1, \ldots, d\}^2$
consisting of couples $(i,j)$ such that $i > \sigma^{-1}(j)$. One
can check that it has cardinality $\frac{d(d-1)} 2$. Consider the 
function $f_x : \Omega(x) \to R^{I_x}$ mapping $\omega$ to the family 
$(\omega_{i,j})_{(i,j) \in I_x}$. Following the construction of the
$X_{i,j}$'s, one can check that $f_x$ is a bijection.

Now, we globalize the previous construction. Let $U$ be a subset of 
$R^I$ containing a distinguished element $x$ and such that $I_y = I_x$ 
for all $y \in U$. With this assumption, the collection of functions 
$f_y$'s ($y$ varying in $U$) defines a bijection between $\Omega(U) 
= \{ \, \omega \in \Omega \, | \, (X_{i,j}(\omega))_{(i,j) \in I} \in U 
\, \}$ and $U \times R^{I_x}$. It is morever easy to check that this 
bijection preserves the measure; in other words
\begin{equation}
\label{eq:XijU}
\Prob[(X_{i,j})_{(i,j) \in I} \in U] = \mu(U)
\end{equation}
where $\mu$ denotes the Haar measure on $R^I$. But, since the function 
$v_R$ is locally constant on $R \backslash \{0\}$, any open subset $U
\subset (R \backslash \{0\})^I$ can be written as a disjoint union of
subsets $U'$ on which the function $y \mapsto I_y$ is constant. 
Therefore the equality \eqref{eq:XijU} holds for all these $U$. Since 
furthemore the complement of $(R \backslash \{0\})^I$ in $R^I$ is a 
measure-zero set, the equality \eqref{eq:XijU} holds for all open
subset $U$ of $R^I$.
\end{proof}

\begin{cor}
\label{cor:Vij}
The random variables $V_{i,j}$ ($1 \leq i \leq j \leq d$) are
mutually independent and they all follow a geometric law of parameter 
$1 - q^{-1}$ (\emph{i.e.} they take value $v$ with probability $(1-q) 
q^{v-1}$).
\end{cor}

\begin{proof}
Clear after Proposition \ref{prop:Xij}.
\end{proof}

Another interest of the $V_{i,j}$'s is that they are closely related to 
$V_L$. The following Proposition precises this relationship.

\begin{prop}
\label{prop:Vij}
We have $\max (V_1, V_2, \ldots, V_d) - v_R(\det) \leq V_L \leq \max 
(V_1, \ldots, V_{d-1})$ (recall that $V_i = V_{i,i}$ by definition).
\end{prop}

\begin{proof}
Let $\omega \in \Omega$. To avoid confusion, agree to call $T_j(\omega)$ 
the matrix $\omega$ computed by Algorithm \ref{algo:Vij} (run with $\omega$ 
as input) after the $j$-th iteration of the main loop and reserv the 
notation $\omega$ for the matrix we have started with. It follows from the 
construction that $T_j(\omega)$ has the following particular shape: if 
$i' < j' \leq j$, then the $(i',j')$-th entry of $\omega_j$ vanishes.
Moreover, clearly, $T_j(\omega)$ is obtained from 
$\omega$ by performing successive elementary operations on the first $j$ 
columns. Therefore, if $J = \{1, \ldots, j\}$ and if $I$ is a subset of 
$\{1, \ldots, d\}$ of cardinality $J$, we have $\det \omega_{I,J} = \pm 
\det (\omega_j)_{I,J}$. In particular these two determinants have the 
same valuation. Fix a couple $(i,j)$ such that $1 \leq j \leq
i \leq d$ and set $I = \{1, \ldots, j-1, i\}$, $J = \{1, \ldots, j\}$.
From Formula \eqref{eq:CramerL} and what we have said before, we derive:
\begin{eqnarray*}
v_R(L(\omega)_{i,j}) & = & 
v_R(\det T_j(\omega)_{I,J}) - v_R(\det T_j(\omega)_{J,J}) \\
& = & 
v_R(T_j(\omega)_{i,j}) - v_R(T_j(\omega)_{j,j}) =
v_R(T_j(\omega)_{i,j}) - V_j(\omega).
\end{eqnarray*}
Since all coefficients of $\omega_j$ lie in $R$, so does its
determinant. It follows that $v_R(T_j(\omega)_{i,j}) \geq 0$ and
consequently that $v_R(L(\omega)_{i,j}) \geq - V_j(\omega)$, which
proves the second inequality. To establish the first one, note that
$\omega$ and $T_j(\omega)$ share the same determinant up to a sign. Thus
there must exist an index $i$, necessarily not less than $j$, such 
that $v_R(T_j(\omega)_{i,j}) \leq v_R(\det \omega)$. 
For this particular $i$, we have $v_R(L(\omega)_{i,j}) \leq v_R(\det 
\omega) - V_j(\omega)$ and then $V_L(\omega) \geq V_j(\omega) - 
v_R(\det \omega)$. The conclusion follows.
\end{proof}

\begin{rem}
\label{rem:det}
In the same way, we can prove that the valuation of the $i$-th minor
of $\omega \in \Omega$ is equal to
$\sum_{i=1}^j \min(V_{i,i}(\omega), V_{i,i+1}(\omega), \ldots, 
V_{i,j}(\omega))$.
Combining this with Corollary \ref{cor:Vij}, one can easily recover 
Abdel-Ghaffar's formula $\sum_{i=1}^d \frac 1 {q^i-1}$ (see Theorem 3
of \cite{abdel}) giving the expected value of the random variable 
``valuation of the determinant''.
\end{rem}

\subsection{Proof of the main results}
\label{subsec:proofstats}

\subsubsection{Estimation of the expected value}
\label{subsec:expvalue}

This subsection is devoted to the proof of Theorem \ref{th:VLexp}.

\paragraph{Estimation of the expected value of $V_L$}

Let $V = \max(V_1, V_2, \ldots, V_d)$. The event ``$V < v$'' occurs 
if and only if $V_{i,i} < v$ for all index $i$, and Corollary 
\ref{cor:Vij} shows that it happens with probability $(1 - q^{-v})^d$. 
The expected value of $V$ is then equal to $\sum_{v=1}^\infty \Prob[V 
\geq v] = \sum_{v=1}^\infty \big[ 1 - (1 - q^{-v})^d \big]$ that is 
exactly $E(q,v)$. On the other hand, Proposition \ref{prop:Vij} implies 
that $\E[V] - \E[v_R(\det)] \leq \E[V_L] \leq \E[V]$. Moreover, by
Abdel-Ghaffar's Theorem, we know that the expected value of $v_R(\det)$
is given by $\sum_{i=1}^d \frac 1{q^i - 1}$ and hence is less that
$\sum_{i=1}^d \frac 1{q^i} < \sum_{i=1}^\infty \frac 1{q^i} = \frac 
1{q-1}$. The first part of Theorem \ref{th:VLexp} is proved.

\paragraph{Estimation of $E(q,v)$}

Consider the function $f : x \mapsto 1 - (1-q^{-x})^d$. It is decreasing
on the interval $[0, \infty)$ and therefore one can write:
$$\int_0^\infty f(x) dx \geq E(q,d) \geq 
\int_1^\infty f(x) dx \geq -1 + \int_0^\infty f(x) dx.$$
Doing the substitution $y = 1-q^{-x}$, we get:
\begin{eqnarray*}
\int_0^\infty f(x) dx = \frac 1 {\log q} \cdot \int_0^1 
\frac{1-y^d}{1-y} \: dy = 
\frac 1 {\log q} \cdot \int_0^1 (1 + y + y^2 + \cdots + y^{d-1}) dy
= \frac {H_d} {\log q}
\end{eqnarray*}
where $H_d = 1 + \frac 1 2 + \cdots + \frac 1 d$ is the harmonic series. 
It is well known that $\gamma + \log d \leq H_d \leq 1 + \log d$ where 
$\gamma$ is the Euler's constant. Therefore $E(q,d)$ is almost equal to 
$\log_q d$, the error term being bounded by a universal constant. The 
second part of Theorem \ref{th:VLexp} follows.

\paragraph{Some additional remarks}

We would like first to emphasize that the difference $E(q,d) - \log_q d$ 
does \emph{not} converge to $0$ when $q$ and/or $d$ goes to infinity. 
Indeed the following Lemma shows that, when $\log_q d$ is far from an 
integer and $q$ is large, $E(q,d)$ might be closer to the integral part 
of $\log_q d$ than to $\log_q d$ itself.

\begin{lemma}
For all $q$ qnd $d$, 
$$\big| E(q,d) - [\log_q d] \big| < \frac q{q-1} \cdot
q^{-\dist(\log_q d, \, \N)}$$
where $[\log_q d]$ and $\dist(\log_q d, \N)$ denotes respectively
the integral part and the distance to $\N$ of $\log_q d$.
\end{lemma}

\begin{proof}
We claim that the function $f : x \mapsto 1 - (1-q^{-x})^d$ satisfies:
\begin{equation}
\label{eq:encadf}
1 - \frac {q^x} d \leq f(x) \leq \frac d {q^x}, \quad
\text{for all } x \geq 0.
\end{equation}
Indeed, the second inequality directly comes from the standard 
inequality $(1+t)^d \leq 1 + td$ whereas the first one is a consequence 
of AM-GM inequality applied with the numbers $d q^{-x}$ and $1-q^{-x}, 
1-q^{-x}, \ldots, 1-q^{-x}$ ($d$ times). If $v_0 = [\log_q d]$, we then 
get $v_0 + \sum_{v=1}^{v_0} \frac {q^v} d \leq E(q,d) \leq v_0 + 
\sum_{v=v_0+1} ^\infty \frac d {q^v}$, which gives:
$$- \frac q {q-1} \cdot \frac{q^{v_0}} d \leq E(q,d) - v_0 \leq
\frac q {q-1} \cdot \frac d {q^{v_0+1}}.$$
The Lemma follows from this.
\end{proof}

Let us end this paragraph by a last remark: the sum $E(q,d)$ can also
be exactly computed. Indeed, we have:
$$E(q,d) = \sum_{v=1}^\infty 1 - (1-q^{-v})^d = \sum_{v=1}^\infty 
\sum_{k=1}^d (-1)^{k-1} \binom d k q^{-vk} = \sum_{k=1}^d (-1)^{k-1} 
\binom d k \cdot \frac 1 {q^k-1}.$$
Nevertheless, this expression does not yield the order of magnitude of 
$E(q,d)$; indeed, each term in the latter sum (the one over $k$) can 
individualy be very large whereas the sum itself grows rather slowly.

\subsubsection{Estimation of the law of $V_L$}

We now start the proof of Theorem \ref{th:VLlaw}.
The strategy is quite clear: we use Corollary \ref{cor:Vij} and 
Proposition \ref{prop:Vij} to bound from below and from above the 
distribution function of $V_L$.
First, let us investigate the consequences of the inequality $V_L \leq 
V$ (where we recall that we have set $V = \max(V_1, \ldots, V_d)$). For 
all (nonnegative) real number $x$, it implies that:
\begin{equation}
\label{eq:upbdist}
\Prob[V_L < x] \geq \Prob[V < x] = \prod_{i=1}^d \Prob[V_i <
x] \geq (1 - q^{-x})^d \geq 1 - d \cdot q^{-x}.
\end{equation}
It is a bit more tricky to use the other inequality $V_L \geq V - 
v_R(\det)$ because $v_R(\det)$ and the $V_i$'s are certainly not 
independent (\emph{cf} Remark \ref{rem:det}). Nevertheless, one can 
pick two nonnegative real numbers $x$ and $t$ and consider the
event $E_{x,t}$ : ``$V > x+t$ and $v_R(\det) \leq t$''. It is clear 
that $V - v_R(\det)$ is always greater than $x$ when $E_{x,t}$ occurs. 
Thus we have:
\begin{equation}
\label{eq:lowbdisti}
\Prob[V_L \leq x] \leq \Prob[V - v_R(\det) \leq x] \leq 
\Prob[E_{x,t}] \leq 1 - \Prob[V \leq x+t] - \Prob[v_R(\det) > t].
\end{equation}
Moreover we know that $\Prob[V \leq x+t] \leq (1 - q^{-x-t})^d$ and
from Abdel-Ghaffar's result (see \cite{abdel}), we derive $\Prob[V_d > 
t] \leq \frac {q^{-t+2}} {q-1}$. Indeed, Abdel-Ghaffar Theorem states 
that for all integer $v$, the equality $\Prob[v_R(\det) \leq v] = (1 - 
q^{-v-1}) (1 - q^{-v-2}) \cdots (1 - q^{-v-d})$ holds. In particular 
$\Prob[v_R(\det) \leq v] \geq 1 - \sum_{i=1}^d q^{-v-i} \geq 1 - 
\frac{q^{-v}}{q-1}$. Taking $v = [t]$, we get the claimed result. 
Putting these inputs in \eqref{eq:lowbdisti}, we obtain:
$$\Prob[V_L \leq x] \leq 1 - (1 - q^{-x-t})^d - \frac {q^{-t+2}}{q-1}.$$
This estimation being true for all $t$, one can optimize it on $t$.
For simplicity, let us define $u = 1 - q^{-x-t}$; the variable $u$ now
varies in $[1-q^{-x}, 1]$, and for all $u$ in this range, one have 
$\Prob[V_L > v] \geq 1 - f(u)$ where $f(u) = u^d + d \lambda (1-u)$, 
$\lambda = \frac{q^{x+2}} {d(q-1)}$. Assume that $\lambda < 1$. A quick 
study of $f$ shows that it is minimal when $u = u_0 = \lambda^{1/(d-1)}$. 
Moreover, one can check (using AG-MG inequality for instance) that $u_0$ 
always lies in the interval $[1-q^{-x}, 1]$. It follows that $\Prob[V_L 
\leq x] \geq 1 - f(u_0) = 1 - \lambda \cdot (d - (d-1) u_0)$. We can 
further simplify this formula and write a bound depending only on 
$\lambda$. For this, remark that $\lambda \geq (1 + \frac{\log 
\lambda}{d-1})^{d-1}$. Raising to the power $d-1$, we find $u_0 \geq 
1 + \frac{\log \lambda}{d-1}$ and then:
\begin{equation}
\label{eq:lowbdist}
\Prob[V_L \leq x] \leq \lambda (1 - \log \lambda)
\quad \text{where }
\lambda = \frac{q^{x+2}} {d(q-1)}.
\end{equation}
We are now ready to prove Theorem \ref{th:VLlaw}. Let $\ell$ be a 
positive real number and define $v_0 = \log_q d - \frac 1 2$. Applying 
Formulas \eqref{eq:upbdist} and \eqref{eq:lowbdist} with $x = v_0 + 
(\ell + \frac 1 2)$ and $x = v_0 - (\ell + \frac 1 2)$ respectively, we 
find:
\begin{eqnarray*}
\Prob[V_L \geq v_0 + (\ell + {\textstyle \frac 1 2)}] & \leq & q^{-\ell} 
\\
\text{and} \quad
\Prob[V_L \leq v_0 - (\ell + {\textstyle \frac 1 2})] & \leq & 
\frac q{q-1} \cdot q^{-\ell} \cdot \Big(1 - \log\Big(\frac q{q-1}\Big) + 
\ell \cdot \log q\Big) \\
& \leq & \frac q{q-1} \cdot q^{-\ell} \cdot (1 + \ell \cdot \log q).
\end{eqnarray*}
Theorem \ref{th:VLlaw} follows by adding these two inequalities.
Corollary \ref{cor:VLdev} can be now easily deduced. Indeed, note 
that the function $v \mapsto \E((V_L - v)^2)$ is maximal when $v$ is 
equal to the expected value of $V_L$ and the value taken at this optimal 
point is the variance of $V_L$. It is then enough to bound the expected
value of $(V_L - v_0)^2$, which can be done as follows:
\begin{eqnarray*}
\E[(V_L - v_0)^2] & = & \int_0^\infty \Prob\big[(V_L - v_0)^2 \geq 
x\big] \cdot dx \\
& \leq & \frac 1 4 + \int_0^\infty \Prob\big[(V_L - v_0)^2
\geq {\textstyle (\ell + \frac 1 2)^2}\big] \cdot (2\ell+1) \cdot 
d\ell \\
& \leq & \frac 1 4 + \frac q {q-1} \cdot \int_0^\infty  q^{-\ell} \cdot 
(2 + \ell \cdot \log q) \cdot (2\ell+1) \cdot d\ell \\
& = & \frac 1 4 + \frac q{q-1} \cdot \Big( \frac 3{\log q} + \frac 8
{\log^2 q}\Big).
\end{eqnarray*}
The standard deviation of $V_L$ is then always less than $\sigma(q) = 
\sqrt{ \frac 1 4 + \frac q{q-1} \cdot ( \frac 3{\log q} + \frac 8 
{\log^2 q})}$. The function $\sigma$ is decreasing on $[2, \infty)$ and 
then bounded from above by its value at $2$ (which is $< 6.5$). Note 
furthemore that when $q$ goes to infinity, $\sigma(q) = \frac 1 2 + 
O(\frac 1{\log q})$.

\subsection{Generalization to block LU decomposition}
\label{subsec:block}

Let $\partd = (d_1, \ldots, d_r)$ be a tuple of positive integers such 
that $d_1 + \cdots + d_r = d$. By definition, a block LU decomposition 
of type $\partd$ of a matrix $M \in M_d(K)$ is a factorization $M = 
L_{\partd}(M) \cdot U_{\partd}(M)$ where $L_{\partd}(M)$ and 
$U_{\partd}(M)$ are respectively block unit lower triangular and block 
upper triangular with respect to the partition $\partd$:
$$L_{\partd}(M) = 
\left( \begin{matrix}
I_{d_1} & 0 & \cdots & 0 \\
\star & I_{d_2} & \ddots & \vdots \\
\vdots & \ddots & \ddots & 0 \\
\star & \cdots & \star & I_{d_r}
\end{matrix} \right)
\quad \text{and} \quad
U_{\partd}(M) = 
\left( \begin{matrix}
\star & \cdots & \cdots & \star \\
0 & \ddots &  & \vdots \\
\vdots & \ddots & \ddots & \vdots \\
0 & \cdots & 0 & \star
\end{matrix} \right)$$
where the $s$-th block has size $d_s$ and, for an integer $n$, $I_n$ 
denotes the identity matrix of size $n$.
Of course, a block LU decomposition of type $(1, 1, \ldots, 1)$ is 
nothing but a standard LU decomposition and every matrix $M \in M_d(K)$ 
admits a block LU decomposition of type $(d)$, which is simply $M = I_d 
\cdot M$. As in the standard case, a LU decomposition of type $\partd$ 
is unique (when it exists) --- which justifies the notations $L_{\partd}
(M)$ and $U_{\partd}(M)$ --- and, an invertible matrix $M$ admits such a 
decomposition if and only if, for all $i \in \{1, \ldots, s\}$, its 
$d_i$-th principal minor does not vanish.
For $\partd$ as before and $\omega \in \Omega$, we let $V_{L,\partd}$ 
denote the opposite of the smallest valuation of an entry of $L_{\partd}
(\omega)$. This defines a random variable $V_{L,\partd} : \Omega \to
\N \cup \{\infty\}$ for each $\partd$. The aim of this subsection is
to study them. Following the same strategy as in the standard case 
(\emph{i.e.} $\partd = (1, \ldots, 1)$), our first task is to establish 
a link between $V_{L,\partd}$ and the random variables $V_{i,j}$ defined 
in \S \ref{subsec:Vij}. To shorten notations, we set $I_s = \{d_1 +
\cdots + d_{s-1}+1, \ldots, d_1 + \cdots + d_s\}$ and recall that if
$M \in M_d(K)$ and $I, J \subset \{1, \ldots, d\}$, we denote by 
$M_{I,J}$ the submatrix of $M$ consisting of entries whose row index 
and column index are in $I$ and $J$ respectively. For all $s \in \{1,
\ldots, t\}$, we further introduce:
$$V_{\partd,s} = \sum_{i \in I_s} \: 
\min (V_{i,i}, V_{i,i+1}, \ldots, V_{i,d_1 + \cdots + d_s}).$$
Corollary \ref{cor:Vij} learns us that the $V_{\partd, s}$'s are 
mutually independant for $s$ varying between $1$ et $r$ (and $\partd$ 
remains fixed). The following Lemma shows that their laws are also
precisely known.

\begin{lemma}
\label{lem:Vds}
For all $s \in \{1, \ldots, r\}$ and all integer $v$, we have:
$$\Prob[V_{\partd,s} \leq v] = (1 - q^{-v-1}) (1 - q^{-v-2}) \cdots 
(1 - q^{-v-d_s}).$$
\end{lemma}

\begin{proof}
Throughout this proof, we set $a = d_1 + \cdots + d_{s-1}$, $b = d_1 + 
\ldots + d_s$ and, for $i \in \{1, \ldots, d_s\}$ and $i' = b + 1 - i$,
$W_i = \min(V_{i',i'},V_{i',i'+1}, \ldots, V_{i',b})$. It follows from 
Corollary \ref{cor:Vij} that $W_i$ follows a geometric law of parameter 
$(1-q^{-i})$ and furthermore that the $W_i$'s ($1 \leq i \leq d_s$) are 
mutually independant. For all $\ell \in \{1, \ldots, d_s\}$, define 
moreover $S_\ell = W_1 + \cdots + W_\ell$.
Clearly $S_{d_s} = V_{\partd,s}$. We will prove by induction on the
couple $(\ell,v)$ (lexicographically ordered) that:
$$\Prob[S_\ell \leq v] = (1 - q^{-v-1}) (1 - q^{-v-2}) \cdots
(1 - q^{-v-\ell}).$$
For $\ell = 1$, the statement is true. Assume now that it is true for 
all $(\ell',v')$ with $\ell' < \ell$ or $\ell' = \ell$ and $v' < v$. 
The strategy is to decompose the event ``$S_\ell \leq v$'' in two parts 
according to the vanishing or the nonvanishing of $W_{d_s-\ell+1}$. 
Clearly, if $W_{d_s-\ell+1} = 0$, we have $S_\ell = S_{\ell-1}$. On the 
other hand, if we know for sure that $W_{d_s-\ell+1}$ does not vanish, 
one can subtract $1$ to it and get this way a new random variable which 
still follows of a geometric law with the same parameter. Hence, one can 
write:
\begin{eqnarray*}
\Prob[S_\ell \leq v] & = &
\Prob[W_{d_s-\ell+1} = 0] \cdot \Prob[S_{\ell-1} \leq v] +
\Prob[W_{d_s-\ell+1} > 0] \cdot \Prob[S_\ell \leq v-1] \medskip \\
& = & (1-q^{-\ell}) \cdot \Prob[S_{\ell-1} \leq v] +
q^{-\ell} \cdot \Prob[S_\ell \leq v-1].
\end{eqnarray*}
Replacing $\Prob[S_{\ell-1} \leq v]$ and $\Prob[S_\ell \leq v-1]$ by 
their values (coming from the induction hypothesis), we get the desired 
result.
\end{proof}

\begin{rem}
Alternatively, one can notice that $V_{\partd,s}$ follows the same law 
as the variable ``determinant of a random matrix of size $d_s$'' 
and then conclude by Abdel-Ghaffar's Theorem. Actually the proof we
have presented above is \emph{very} inspired by Abdel-Ghaffar's one.
\end{rem}

\begin{prop}
\label{prop:VLd}
We have $\max(V_{\partd,1}, \ldots, V_{\partd,r}) - v_R(\det) 
\leq V_{L,\partd} \leq \max(V_{\partd,1}, \ldots, V_{\partd,r})$.
\end{prop}

\begin{proof}
We follow the lines of the proof of Proposition \ref{prop:Vij}. To avoid 
confusion, we begin by letting $T_j(\omega)$ denote the matrix $\omega$ 
computed by Algorithm \ref{algo:Vij} (run with $\omega$ as input) after 
the $j$-th iteration of the main loop. Pick some $s \in \{1, \ldots, 
r\}$ and set $j(s) = d_1 + \ldots + d_s$. We are going to prove the two 
following statements from which the Proposition will follow directly:
\begin{itemize}
\item the determinant of the square $d_s \times d_s$ matrix 
$T_{j(s)}(\omega)_{I_s, I_s}$ has valuation $V_{\partd,s}(\omega)$;
\item for all $t \in \{s, \ldots, r\}$, we have the identity 
$L_{\partd}(\omega)_{I_t,I_s} \cdot T_{j(s)}(\omega)_{I_s, I_s} = 
T_{j(s)}(\omega)_{I_t,I_s}$.
\end{itemize}
The first assertion is easily proved. Indeed, by construction, 
the submatrix $T_{j(s)}(\omega)_{\{1, \ldots, j(s)\},\{1, \ldots, 
j(s)\}}$ is lower triangular and that its $i$-th diagonal entry has 
valuation $\min (V_{i,i}, V_{i,i+1}, \ldots, V_{i,j(s)})$. To prove the 
second assertion, we first remark that, up to replacing $\omega$ by 
$\omega + \pi^N$ for a sufficiently large integer $N$, one may assume 
that $\omega$ is invertible. All the matrices $T_{j(s)}(\omega)_{I_s, 
I_s}$ are then also invertible. Consider the matrix $L \in M_d(K)$ whose 
$i$-th column is the $i$-th column of $T_{j(s)}(\omega)$ where $s$ is 
the unique such that $i \in I_s$. It is apparently lower block 
triangular with respect to the partition $\partd$. Furthermore, noting 
that, for $i \in I_s$, the $i$-th column of $T_{j(s)}(\omega)$ is a 
linear combination of the first $j(s)$ columns of $\omega$, we see that 
$L^{-1} \cdot \omega$ is upper block triangular. Hence, if $D$ is the 
diagonal block matrix:
$$D = \left( \begin{matrix}
T_{j(1)}(\omega)_{I_1, I_1} \\
& \ddots \\
& & T_{j(r)}(\omega)_{I_r, I_r}
\end{matrix} \right)$$
the factorization $\omega = (L D^{-1})\cdot (D L^{-1} \omega)$ is the
LU decomposition of type $\partd$ of $\omega$. Therefore $L_{\partd} 
(\omega) = L D^{-1}$. Our claim follows directly from this.
\end{proof}

From Lemma \ref{lem:Vds}, we easily derive that $q^{-v} \leq \Prob[
V_{\partd, s} \geq v] \leq \frac{q}{q-1} \cdot q^{-v}$. Arguing then 
as in \S \ref{subsec:proofstats}, one can prove analogues of the
results we have shown before concerning the random variable $V_L$: the 
expected value of $V_{L,\partd}$ is equal to $\log_q s + O(1)$, its 
standard deviation is a $O(1)$ (where the notation $O(1)$ stands for 
a quantity bounded by a universal constant which can be made explicit)
and, actually, we even have a more precise (but also more technical)
estimation of its law in the spirit of Theorem \ref{th:VLlaw}.

\section{LU decomposition over a DVR: algorithmic issues}
\label{sec:algo}

LU decomposition is a very basic and important tool when we are doing 
algorithmics involving matrices, and especially matrices over a complete 
DVR. But unfortunately, on some particular inputs, computing it may 
cause important numerical instability; it is the case for instance if 
the top left entry of the input matrix has a very large valuation 
(compared to the other entries). The first aim of this second section, 
is to study this phemonemon; more precisely, following the ideas of \S 
\ref{sec:stats}, we will design a new algorithm to compute LU 
decomposition (see Algorithm \ref{algo:stableLU}) and show that the set 
of unpleasant inputs for which it is numerically unstable is very small.

In particular, we may expect that if Algorithm \ref{algo:stableLU} is 
called as a subroutine by an other \emph{probabilistic} algorithm, it 
will not never generate important losses of precision. 
In \S\S \ref{subsec:simulPLU} and \ref{subsec:sheaf}, we will illustrate 
this idea on a particular example: we will propose a probabilistic 
stable algorithm (based on LU decomposition) whose aim is to compute a 
basis of a coherent over $\A^1_K$ (where $K$ is the fraction field of a 
complete DVR) from the knowledge of all its stalks.

\medskip

We keep the general notations of \S \ref{sec:stats}: let $R$ be a 
discrete valuation ring whose valuation $v_R : R \to \N \cup \{\infty\}$ 
is assumed to be surjective. Let $\pi$ be an element of $R$ of valuation 
$1$. Let $k$ (resp. $K$) denote the residue field (resp. the fraction 
field) of $R$ and set $q = \card\: k$. We recall that $v_R$ extends 
uniquely to $K$ and that, in a slight abuse of notations, we continue to 
denote by $v_R$ this extended map. We recall also that we have set 
$\Omega = M_d(R)$ and that this space is endowed with its Haar measure. 
For $\omega \in \Omega$, denote by $W_i(\omega)$ the valuation of the 
$i$-th principal minor of $\omega$ and set $W = \max(W_1, \ldots, W_d)$.
Thanks to Abdel-Ghaffar's Theorem (see \cite{abdel}), the law of the 
$W_i$'s is known: $\Prob[W_i \leq v] = (1 - q^{-v-1}) (1 - q^{-v-2}) 
\cdots (1 - q^{-v-i})$ for all $v \geq 0$ and $i \in \{1, \ldots, d\}$. 
From this, we derive $\Prob[W_i > v] \leq \frac{q^{-v}}{q-1}$ and then:
\begin{equation}
\label{eq:}
\Prob[W > v] \leq d \cdot \frac {q^{-v}}{q-1}
\end{equation}
for all nonnegative integer $v$. Adding all these probabilities, one
finds $\E[W] = \log_q d + O(1)$ where, as usual, the notation $O(1)$ 
refers to a quantity bounded by a universal constant.

\subsection{Loss of precision in LU decomposition}
\label{subsec:stableLU}

By Formula \eqref{eq:CramerL}, we know that the entries of $L(M)$ can be 
all expressed as the quotient of one minor by one principal minor. 
Noting that if $x$ and $y$ are both known with precision $O(\pi^N)$ and 
if $y$ has valuation $v$, the quotient $\frac x y$ is known with 
precision at least $O(\pi^{N-2v})$, one may expect that a good algorithm 
computing the LU factorization of $\omega$ would shrink the initial 
precision by a factor $\pi^{2 \cdot W(M)}$.

Unfortunately, a quick experiment shows that the naive algorithm based
on usual Gauss elimination generates losses of precision much more 
important than that. For example, on a random input matrix $M \in 
M_{25}(\Z_5)$ given with precision $O(5^N)$, it outputs a matrix $L$ 
which is in average known up to precision $O(5^{N-c})$ where $c \simeq 
10$ whereas the mean value of $2 \cdot W(M)$ is only $\simeq 2 \cdot 
\log_q d = 4$. For matrices of size $d = 125$, the deviation is 
amplified: we find $c \simeq 50$... to be compared to $2 \cdot \log_q d 
= 6$.

\subsubsection{A first simple solution}

Our starting remark is the following: it follows from Cramer like 
formulae \eqref{eq:CramerL} that if $M$ are $M'$ are two matrices in
$M_d(R)$ congruent modulo $\pi^N$ (for some positive integer $N$)
such that $W(M) < N$, that $W(M') = W(M)$ and 
$$L_{i,j}(M) \equiv L_{i,j}(M') \pmod {\pi^{N-2 \cdot W_i(M)}}$$
for all $i,j \in \{1, \ldots, d\}$ with $i > j$. In particular, 
under the previous assumptions, we have $L(M) \equiv L(M') \pmod 
{\pi^{N-2\cdot W(M)}}$. This result suggests the following method
to compute $L(M)$ with a correct precision when $M$ is a matrix known
with precision $O(\pi^N)$:
\begin{itemize}
\item we lift $M$ to a matrix $M'$ known with precision $O(\pi^{N'})$
for some $N' > N$;
\item we compute $W(M')$ and $L(M')$ with our favorite algorithm
(\emph{e.g.} Gauss elimination)\footnote{Generally, these two 
computations can be done simultaneously. It happens in particular
if one uses Gauss elimination.};
\item we answer $L(M) = L(M') + O(\pi^{N-2\cdot W(M')})$.
\end{itemize}
By what we have said before, our answer $L(M)$ is always correct. 
Furthemore, if $N'$ is sufficiently large, then $L(M')$ will be known 
with precision at least $O(\pi^{N-2\cdot W(M')})$ and $L(M)$ itself
will be known with precision $O(\pi^{N-2\cdot W(M')})$.

It then remains to find a suitable value for $N'$. Of course, it will 
strongly depend on the algorithm we use to compute $L(M')$. Let us study 
a bit the case of Gauss elimination. Since the successive pivots 
appearing during the elimination have valuations $W_1(M'), W_2(M'), 
\ldots, W_d(M')$ and since we are only dividing by pivots, the maximal 
loss of precision is bounded from above by $2\cdot (W_1(M') + \cdots + 
W_d(M'))$. In other terms, using Gauss elimination, one can certainly 
compute $L(M)$ with precision $O(\pi^{N-2\cdot (W_1(M') + \cdots + 
W_d(M'))})$. As a consequence, it is enough to choose $N'$ so that:
$$N' - N \geq 2 \cdot \big(W_1(M) + \cdots + W_d(M) - W(M)\big).$$
However, at the very beginning, we have not computed the $W_i(M)$'s 
yet. So we cannot figure out at this moment what is the best value of 
$N'$ (\emph{i.e.} the smallest one satisfying the above inequality). 
Nevertheless, we know that in average $W_i(M') \simeq \frac 1 q$ and 
$W(M) \simeq \log_q d$. To begin with, we can then try to take $N' = N + 
\lceil\frac {2d} q\rceil$ and see what happens: we do the computation with this 
particular $N'$, we determine the $W_i(M)$'s, if the above inequality is 
fulfilled, we are done, otherwise, we determine the right $N'$ and redo 
the computation. Actually, it could happen --- but it is very rare ---
that the first precision $O(\pi^{N'})$ does not allow us to determine 
some of the $W_i(M)$'s; in that case, we just guess a new larger $N'$, 
try with it and continue like this until it works.

Let us finally analyze the complexity of this method in the favorable 
case where $N' = N + \frac{2d} q$ is enough. In order to fix notations, 
we assume moreover that doing basic operations (\emph{i.e.} additions, 
substractions, multiplications and divisions) in $R$ with precision 
$\pi^N$ requires $O((N \log q)^\alpha)$ bit operations where $\alpha$ is 
some constant\footnote{In usual situations, one can take $\alpha = 1 + 
\varepsilon$ for all positive real number $\varepsilon$.}, necessarily 
greater than or equal to $1$. Since the complexity of Gauss elimination 
is $O(d^3)$ operations in the base ring, our method needs:
$$\textstyle O\big(\:d^3 \cdot (N + \frac d q)^\alpha \cdot \log^\alpha 
q\:\big)$$
bit operations. If $d \ll qN$, it is quite nice. However, if the opposite 
situation when $d \gg qN$, the dominant term in the above complexity is 
$d^{3+\alpha}$, which is very large and actually not really acceptable 
for many practical applications.

\subsubsection{A stable algorithm to compute LU decomposition}
\label{subsubsec:stableLU}

In this subsection, we propose and study a different method to compute 
LU decomposition which has the advantage of not requiring to increase 
the precision at any time and whose complexity is comparable to Gauss 
elimination. Our algorithm is strongly inspired by the constructions of 
\S \ref{sec:stats} and especially those of \S \ref{subsec:Vij}. Here is
it:

\begin{algorithm}
  \SetKwInOut{myInput}{\it Input:}
  \SetKwInOut{myOutput}{\it Output:}
  \SetKwInOut{Nota}{\it Notations:}
  \SetKwInOut{Void}{}

  \myInput{A matrix $M$ of size $d \times d$ known with precision $O(\pi^N)$}
  \myOutput{The $L$-part of the LU decomposition of $M$}

  \BlankLine

  \Nota{$\star\,$ $d$ is the dimension of the matrix $M$}
  \Void{$\star\,$ $A_{i,j}$ denotes the $(i,j)$-th entry of a matrix $A$}
  \Void{$\star\,$ $A_j$ denotes the $j$-th row of $A$}

  \BlankLine

  $\omega$ $\leftarrow$ $M$\;
  $L$ $\leftarrow$ identity matrix of size $d \times d$\;
  \For{$j$ from $1$ to $d$}{
    \For{$i$ from $1$ to $j-1$}{
      \lIf{$v_R(\omega_{i,j}) < v_R(\omega_{i,i})$}
        {swap $\omega_j$ and $\omega_i$\;}
      \lIf {$\omega_{i,i} \neq 0$}
        {$s$ $\leftarrow$ $\frac{\omega_{i,j}}{\omega_{i,i}}$ 
         lifted to precision $O(\pi^N)$;\,\label{line:lift}
         $\omega_j$ $\leftarrow$ $\omega_j - s \cdot \omega_i$%
         \label{line:updateomega}\;}
    }
    $v$ $\leftarrow$ $\sum_{k=1}^j v_R(\omega_{k,k})$\;
    \lFor{$i$ from $j+1$ to $d$}{
      $L_{i,j}$ $\leftarrow$ $\frac{\omega_{i,j}}{\omega_{j,j}}
      + O(\pi^{N-v-\max(0,v_R(\omega_{j,j})-v_R(\omega_{i,j}))})$%
      \label{line:updateL}\;
    }
  }
  \Return{$L$};
\caption{A stable algorithm to compute the $L$-part of the LU
decomposition\label{algo:stableLU}}
\end{algorithm}

A first important remark related to Algorithm \ref{algo:stableLU} is the 
following: at each step, all entries of $\omega$ are known with 
precision $O(\pi^N)$. Indeed, $\omega$ itself is updated only on line 
\ref{line:updateomega} and the corresponding computation does not affect 
the precision (because $s$ has been lifted modulo $\pi^N$ previously).

\paragraph{Correctness of Algorithm \ref{algo:stableLU}}

We fix an integer $j \in \{1, \ldots, d\}$ and focus on the matrix 
$\omega$ computed by the algorithm after the $j$-th iteration of the 
main loop. It is clear that it is obtained from $M$ by performing a 
sequence of elementary operations on its $j$ first columns. Thus, for 
all $i > j$, we have $L(M)_{i,j} = \frac{\det \omega^{(j)}_{I,J}}{\det 
\omega^{(j)}_{J,J}}$ where $I = \{1, \ldots, j-1, i\}$ and $J = \{1, 
\ldots, j\}$.
On the other hand, by construction, $\omega^{(j)}_{I,J}$ and 
$\omega^{(j)}_{J,J}$ are two upper triangular matrices \emph{modulo 
$\pi^N$}. Their determinants are then congruent to the product of their
diagonal entries modulo $\pi^N$. Therefore:
$$L(M)_{i,j} = \frac{\omega^{(j)}_{1,1} \cdots \omega^{(j)}_{j-1,j-1}
\cdot \omega^{(j)}_{i,j}+ O(\pi^N)}{\omega^{(j)}_{1,1} \cdots 
\omega^{(j)}_{j-1,j-1} \cdot \omega^{(j)}_{j,j}+ O(\pi^N)}.$$
Of course, the value of this quotient is $\frac{\omega^{(j)}_{i,j}}
{\omega^{(j)}_{j,j}}$ up to some precision. To compute this precision,
it is easier to work with relative precision (\emph{i.e.} the difference 
between the absolute precision and the valuation); indeed, we know that 
the relative precision of a quotient is equal to the minimum between the
relative precisions of the numerator and the numerator. In our case, if 
we set $v = v_R(\omega^{(j)}_{1,1}) + \cdots + v_R(\omega^{(j)}_{j,j})$ 
and $w = v_R(\omega^{(j)}_{i,j}) - v_R(\omega^{(j)}_{j,j})$, the 
relative precision of the numerator (resp. the denominator) is $N-(v+w)$
(resp. $N-v$). Thus, the relative precision of the quotient is 
$N-v-\max(0,w)$ and its absolute precision is then $N-v+\min(0,w)$
(since its valuation is $w$).
The value $L_{i,j}$ computed by Algorithm \ref{algo:stableLU}, together
with its precision, are then correct.

\paragraph{Precision issues}

Keeping the previous notations, one certainly have $w \geq -v$ and then 
$N-v+\min(0,w) \geq N-2v$. In other words, the $(i,j)$-th entry of the 
matrix $L$ returned by the Algorithm \ref{algo:stableLU} is known with 
precision at least $O(p^{N-2V_j(M)})$ (recall that $V_j(M)$ denotes the 
valuation of $\omega_{j,j}$ at the end of the $j$-th loop, \emph{i.e.} 
our previous $v$). The maximal loss of precision is then bounded above 
by $2 \cdot \max(V_1(M), \ldots, V_d(M))$. By the results of \S 
\ref{sec:stats}, we know that the mean of this upper bound is close to 
$2 \cdot \log_q d$, that is the value we expected.

\subsubsection{Algorithm \ref{algo:stableLU} and Hermite normal form}
\label{subsec:hermite}

Let us denote by $H'(M)$ the matrix $\omega$ computed at the end of the 
execution of Algorithm \ref{algo:stableLU}. It worths remarking that 
$H'(M)$ has a lot of things to do with the Hermite normal form of $M$. 
Let us first agree on the definition of the Hermite normal form of $M$: 
throughout this paper, it will refer to the unique lower triangular 
matrix whose diagonal entries are powers of $\pi$ and which is 
right-equivalent to $M$ (it means that $H(M)$ is obtained from $M$ by 
multiplying on a right by a unimodular matrix). We will denote it by 
$H(M)$.

\begin{prop}
\label{prop:hermite}
Let $M \in M_d(R)$ known with precision $\pi^N$.
We assume that all diagonal entries of $H'(M)$ are not congruent to 
$0$ modulo $\pi^N$ and, for all $j \in \{1, \ldots, d\}$, we write 
$H'_{j,j}(M) = p^{v_j} u_j$ where $v_j$ is a nonnegative integer and 
$u_j$ is a unit. For all $i,j \in \{1, \ldots, d\}$, we then have:
$$\begin{array}{rrcl}
\text{if } i < j: & 
H_{i,j}(M) & = & 0 \\
\text{if } i = j: & 
H_{i,j}(M) & = & \pi^{v_j} \\
\text{if } i > j: & 
H_{i,j}(M) & \equiv & u_j^{-1} \cdot H'_{i,j}(M) \pmod{\pi^{N-v_j}}.
\end{array}$$
\end{prop}

\begin{rem}
Keeping the notations of the Proposition, it is clear that $u_j$ is only 
known modulo $\pi^{N-v_j}$. The congruence of the Proposition is then,
by nature, the best one can expect.
\end{rem}

\begin{proof}
Let $W'(M)$ be the matrix $W'$ computed by Algorithm \ref{algo:stableLU}.
One can easily check that $W'(M)$ is unimodular and moreover that $H'(M)
= M \cdot W'(M)$. Consequently the Hermite normal form of $M$ is equal
to the Hermite normal form of $H'(M)$.

On the other hand, we know that $H'(M)$ has a very particular shape: 
firstly, it is lower triangular modulo $\pi^N$ and secondly, by 
assumption, its diagonal entries are not divisible by $\pi^N$. Thus, 
$H(M)$ is obtained from $H'(M)$ by clearing one by one its entries lying 
above the diagonal and by dividing its $j$-th column by $u_j$. But, if 
$H_{i,j}(M) = \pi^N v_{i,j}$ (for some pair $(i,j)$ with $i<j$), one 
clears the $(i,j)$-th entry of $H'(M)$ by doing the following elementary 
operation on columns: $H'_j(M) \leftarrow H'_j(M) - \pi^{N-v} u_i^{-1} 
v_{i,j} H'_i(M)$. Hence clearings do not affect the value of $H'_{i,j} 
(M)$ modulo $\pi^{N-v_j}$. The Proposition follows easily from this 
observation.
\end{proof}

\subsubsection{The notion of L'V' decomposition}

The $L$-part of the LU decomposition has of course very nice abstract 
properties but unfortunately does not behave very well regarding to 
precision. Indeed, as we have seen before, if a matrix $M$ is known 
modulo $\pi^N$, it is not true that $L(M)$ is known with the same 
precision. But, beyond that, the precision data attached to $L(M)$ is 
not uniform in the sense that all entries of $L(M)$ are \emph{not} known 
with the same precision. In order to tackle this problem, we introduce 
the following definition.

\begin{deftn}
Let $M \in M_d(R)$ and $N$ be a positive integer. A \emph{L'V' 
decomposition of $M$ modulo $\pi^N$} is a couple of $d \times d$ 
matrices $(L',V')$ such that $L' \equiv M V' \pmod {\pi^N}$ and $L'$ and 
$V'$ are lower triangular modulo $\pi^N$ and upper triangular modulo 
$\pi^N$ respectively.

If there exists a diagonal entry of $L'$ which is congruent to $0$ 
modulo $\pi^N$, $(L',V')$ is said to be \emph{degenerate}. Otherwise, 
it is \emph{nondegenerate}.
\end{deftn}

\begin{rem}
It is easy to see that if $(L',V')$ is nondegenerate, then all diagonal 
entries of $V'$ are not congruent to $0$ modulo $\pi^N$ as well.
\end{rem}

It is not difficult to modify Algorithm \ref{algo:stableLU} so that it 
computes a L'V' decomposition modulo $\pi^N$; we end up this way with
Algorithm \ref{algo:LV}.
\begin{algorithm}
  \SetKwInOut{myInput}{\it Input:}
  \SetKwInOut{myOutput}{\it Output:}

  \myInput{A matrix $M$ of size $d \times d$ known with precision $O(\pi^N)$}
  \myOutput{A L'V' decomposition of $M$ modulo $\pi^N$}

  \BlankLine

  $\omega$ $\leftarrow$ $M$\;
  $L',V'$ $\leftarrow$ two new matrices of size $d \times d$\;
  $W'$ $\leftarrow$ identity matrix of size $d \times d$\;
  \For{$j$ from $1$ to $d$}{
    \For{$i$ from $1$ to $j-1$}{
    \lIf{$v_R(\omega_{i,j}) < v_R(\omega_{i,i})$}
      {swap $\omega_j$ and $\omega_i$;\,
       swap $W'_j$ and $W'_i$\;}
    \If {$\omega_{i,i} \neq 0$}
      {$s$ $\leftarrow$ $\frac{\omega_{i,j}}{\omega_{i,i}}$ 
       lifted to precision $O(\pi^N)$\;
       $\omega_j$ $\leftarrow$ $\omega_j - s \cdot \omega_i$;\,
       $W'_j$ $\leftarrow$ $W'_j - s \cdot W'_i$\;}
    }
    $L'_j$ $\leftarrow$ $\omega_j$;\,
    $V'_j$ $\leftarrow$ $W'_j$\label{line:updateLV}\;
  }
  \Return{$L'$, $V'$};
\caption{An algorithm to compute a L'V' decomposition\label{algo:LV}}
\end{algorithm}
On the other hand, it is worth noting that L'V' decomposition is
closely related to LU decomposition. The following proposition makes
this statement precise.

\begin{prop}
\label{prop:LVLU}
Let $M \in M_d(R)$, $N$ be a positive integer and $(L',V')$ be a 
nondegenerate L'V' decomposition of $M$ modulo $\pi^N$. Then $M$ admits
a LU decomposition and for all $(i,j)$ with $1 \leq i < j \leq d$,
one have:
$$L_{i,j}(M) \equiv \frac{L'_{i,j}}{L'_{j,j}} \pmod {\pi^{N - v_j
- \max(0,v_R(L'_{j,j})-v_R(L'_{i,j}))}}$$
with $v_j = \sum_{k=1}^{j-1} v_R(L'_{k,k}) - v_R(V'_{k,k})$.
\end{prop}

\begin{proof}
Left to the reader (the arguments are very similar to those detailed
in \S \ref{subsubsec:stableLU}).
\end{proof}

Of course, executing first Algorithm \ref{algo:LV} and then applying the 
result of Proposition \ref{prop:LVLU} is almost the same than running 
directly Algorithm \ref{algo:stableLU}. Nevertheless splitting Algorithm 
\ref{algo:stableLU} in two parts can be very useful for some 
applications (we will see an example of this in \S \ref{subsec:hafner}) 
because, as we have already said before, the pair $(L',V')$ is generally 
easier to manipulate than $L(M)$ since it carries a flat precision (and, 
in addition, it consists of two integral matrices if $M$ is itself
integral).

\subsubsection{Complexity and Hafner-McCauley's algorithm}
\label{subsec:hafner}

It is easily seen that the asymptotic complexity of Algorithm 
\ref{algo:stableLU} is $O(d^3)$ (operations in the base ring $R$) where 
$d$ denotes the size of the input matrix. It is then similar to the 
complexity of usual Gauss elimination whereas it is true that our 
Algorithm \ref{algo:stableLU} runs a little bit more slowly because it 
basically makes more swaps and copies.

When precision is not an issue (\emph{e.g.} when we are working over an 
exact ring), Hafner and McCauley showed in \cite{hafner} how to reduce 
the computation of the LU decomposition to matrix multiplication and got 
this way a nice recursive algorithm that computes the LU decomposition 
of a matrix in only $O(d^\omega)$ operations where $\omega$ is the 
exponent for matrix multiplication\footnote{Nowadays, the best known 
value for $\omega$ is $2.376$ but, unfortunalety, the corresponding 
algorithm due to Coppersmith and Winograd (see \cite{coppersmith}) is 
not efficient in practice (even for very large $d$) because the constant 
hidden in the $O$ is quite large. A good compromise is to use classical 
Strassen's algorithm 
whose asymptotic complexity 
is a little bit worse --- exactly $O(d^{\log_2 7}$) --- but which is 
easy to implement and works very well in practice.}.
The aim of this subsection is to extend Hafner-McCauley's algorithm in 
our setting where we want to take care of precision.

\paragraph{A preliminary result about Algorithm \ref{algo:LV}}

Roughly speaking, Algorithm \ref{algo:LV} clears the entries of $\omega$ 
lying above the diagonal in the colexigographic order. We would like to 
study what happens if we decide to clear these entries in a different 
order.

\begin{deftn}
\label{def:nice}
Let $a < b$ be two positive integers and set $I_{a,b} = \{ (i,j) \in 
\N^2 \, | \, a \leq i < j \leq b \}$.
A total order $\preccurlyeq$ on $I_{a,b}$ is \emph{nice} if:
\begin{itemize}
\item for $1 \leq i \leq i' < j \leq d$, one always have $(i,j) 
\preccurlyeq (i',j)$, and
\item for all $(i,j)$ and $(i',j') \in I_{a,b}$ such that $j \leq i'$, 
one have $(i,j) \preccurlyeq (i',j')$.
\end{itemize}
\end{deftn}

\begin{rem}
\label{rem:nicerec}
It is easy to check that the colexicographic order on $I_{a,b}$ is nice. 
However, it is not the only one: the lexicographic order, for instance,
is nice as well. 
One can also build recursively nice orders on $I_{a,b}$ as follows. 
Fix an integer $c$ between $a$ and $b$ and pick $\preccurlyeq_1$ and 
$\preccurlyeq_2$ two nice orders defined on $I_{a,c}$ and $I_{c+1,b}$ 
respectively. Consider also a third order $\preccurlyeq_3$ defined
on the cartesian product $\{a, \ldots, c\} \times \{c+1, \ldots, b\}$
and satisfying the first condition of Definition \ref{def:nice}. 
Now define a new order $\preccurlyeq$ on $I_{a,b}$ by agreeing that 
$I_{a,c} \preccurlyeq \{a, \ldots, c\} \times \{c+1, \ldots, b\} 
\preccurlyeq I_{c+1,b}$\footnote{By this inequality, we mean that 
elements in $I_{a,c}$ are all less than those in $\{a, \ldots, c\} \times 
\{c+1, \ldots, b\}$ and, in the same way, that the latter elements are 
less than any pair in $I_{c+1,b}$.} and furthermore that $\preccurlyeq$ 
agrees with $\preccurlyeq_1$, $\preccurlyeq_2$ and $\preccurlyeq_3$ on 
$I_{a,c}$, $I_{c+1,b}$ and $\{a, \ldots, c\} \times \{c+1, \ldots, b\}$ 
respectively. A quick check then shows that $\preccurlyeq$ is nice as
well.
\end{rem}

If $\preccurlyeq$ is a nice order on $I_{1,d} = \{ (i,j) \in \N^2 \, | 
\, 1 \leq i < j \leq d \}$, let us agree to use the expression ``to 
execute Algorithm \ref{algo:LV} with respect to $\preccurlyeq$'' to 
mean that we execute this algorithm but, instead of running through 
all $(i,j) \in I_{1,d}$ according to the colexicographic order, we 
run through these pairs according to $\preccurlyeq$ and execute line 
\ref{line:updateLV} when $i = j-1$.

\begin{prop}
\label{prop:niceorders}
When they are called on the same input, Algorithm \ref{algo:LV} and
Algorithm \ref{algo:LV} executed with respect to a nice order return
the same answer.
\end{prop}

\begin{proof}
Easy check.
\end{proof}

\paragraph{Description of the algorithm}

Suppose that we are given a matrix $M \in M_d(R)$ known with precision 
$O(\pi^N)$. The basic idea (which comes from Hafner and McCauley) is to 
obtain a recursive algorithm to compute the LU decomposition and, doing 
this, to replace as much as possible elementary operations on rows by 
matrix multiplication. Moreover, in order to avoid many problems related 
to precision, it would be really better to work with L'V' decomposition 
instead of LU decomposition. 
Actually, for the purpose of the recursion, we will not just need the 
matrices $L'$ and $V'$ but also $H'$ (which is the matrix $\omega$ at
the end of the execution; see \S \ref{subsec:hermite}) and $W'$. The 
prototype of the algorithm we want to design is then:
$${\tt LV} : M \mapsto (L', V', H', W').$$
Proposition \ref{prop:niceorders}, together with the recursive 
construction of a nice order detailed in Remark \ref{rem:nicerec}
suggests the following strategy for a recursive implementation of
{\tt LV}:
\begin{enumerate}
\item we start the computation of a L'V' decomposition of $M$ but 
stop it after $d'$ columns for some $d' < d$ (\emph{e.g.} $d' = 
[\frac d 2]$);
\item we clear all the entries in the $d' \times (d-d')$ top right 
corner of the matrix $\omega$ we have ended up after the first step;
\item we finally compute a L'V' decomposition of the $(d-d') \times 
(d-d')$ bottom right corner of $\omega$.
\end{enumerate}
It turns out that the first step can be computed in a recursive way.
Precisely, we decompose $M$ as a block matrix 
$$M = \Big(\,\begin{matrix} M_1 & M_2 \\ M_3 & M_4 \end{matrix} \,\Big)
\quad \text{(where } M_1 \text{ has size } d' \times d'\text{)}$$
we call recursively the routine {\tt LV} on the input $M_1$ and then
recover the matrices $L'$, $V'$ and $\omega$ (as they have to be just 
after the first step) using the following formulas:
$$
V' = \Big(\,\begin{matrix} V'_1 & 0 \\ 0 & I \end{matrix} \,\Big)
\quad ; \quad
L' = M V' = 
\Big(\,\begin{matrix} L'_1 & M_2 \\ M_3 \cdot V'_1 & M_4
\end{matrix} \,\Big)
\quad ; \quad
\omega = 
\Big(\,\begin{matrix} H'_1 & M_2 \\ M_3 \cdot W'_1 & M_4
\end{matrix} \,\Big)$$
where the quadruple $(L'_1, V'_1, H'_1, W'_1)$ is the output of the
recursive call of {\tt LV}. Last but not least: remark furthermore
that, proceeding this way, we are replacing elementary operators (on the 
columns of $M_3$) by matrix multiplication (by $V'_1$ and $W'_1$). It is 
exactly the benefit we were looking for!

Let us now focus on step 2. With the notations above, it consists in 
clearing all the entries of $M_2$ (using eventually the diagonal entries 
of $H'_1$ as pivots). Of course, this can be done just by running the 
corresponding part of Algorithm \ref{algo:LV}. Nevertheless, we do not 
want to proceed exactly along these lines but we would like instead to 
use a recursive version of this algorithm in order to take advantage 
again of the complexity of the matrix multiplication. Writing such a 
recursive version is actually very similar to what we have done before.
In order to have more coherent notations, let us rename $H'_1$ and $M_2$ 
to $X$ and $Y$ respectively and write:
$$X = \Big(\,\begin{matrix} X_1 & 0 \\ X_3 & X_4 \end{matrix} \,\Big)
\quad ; \quad
Y = \Big(\,\begin{matrix} Y_1 & Y_2 \\ Y_3 & Y_4 \end{matrix} \,\Big).$$
Note that $X_1$ and $X_4$ are then lower triangular modulo $\pi^N$. We
can then proceed recursively along the following lines:
\begin{enumerate}
\item we clear $Y_1$ using $X_1$ as pivot;
\item we clear $Y_2$ using the new $X_1$ as pivot;
\item we clear $Y_3$ using $X_4$ as pivot;
\item we clear $Y_4$ using the new $X_4$ as pivot.
\end{enumerate}
Each of these steps can be done recursively. As in the previous case, we 
just need to be careful and let our recursive routine return not only
the new matrix $X$ gotten after clearing $Y$ but also the transformation 
matrix $T$ such that $(\,\begin{matrix} X & Y \end{matrix}\,) \cdot T 
\equiv (\,\begin{matrix} X' & 0 \end{matrix}\,)$ where $X'$ is the new 
$X$ mentionned previously. Indeed, this matrix is needed to update $X$ 
and $Y$ after each step.

%

\paragraph{A brief study of complexity}

Let us denote $T'(d)$ the complexity of the clearing algorithm we have 
just described (\emph{i.e.} the number of elementary operations on $R$ 
performed by this algorithm when the size of the input matrices is $d$) 
and by $T(d)$ the complexity of our complete recursive algorithm 
computing a LV decomposition. From the description of these algorithms, 
we find:
\begin{eqnarray}
T(d) & = & 2 \cdot \textstyle T(\frac d 2) + T'(\frac d 2) + 
O(d^\omega) \label{eq:recT} \\
T'(d) & = & 4 \cdot \textstyle T'(\frac d 2) + O(d^\omega) 
\label{eq:recTp}
\end{eqnarray}
where we recall that $\omega$ is the exponant of the complexity of
matrix multiplication. Since a $d \times d$ matrix have $d^2$ entries,
one certainly have $\omega \geq 2$. For simplicity, we assume that
$\omega > 2$ (we recall that the fastest asymptotic algorithm known
today corresponds to $\omega \simeq 2.376$). It is then a classical 
exercise to deduce from the recursion formula \eqref{eq:recTp} that 
$T'(d) = O(d^\omega)$. Knowing this, equation \eqref{eq:recT} becomes 
$T(d) = 2 \cdot T(\frac d 2) + O(d^\omega)$ and then yields $T(d) = 
O(d^\omega)$ as expected.

\subsubsection{Block LU decomposition}

The results of \S \ref{subsubsec:stableLU} extend to block LU decomposition 
using \S \ref{subsec:block}. Indeed, a close look of the proof of 
Proposition \ref{prop:VLd} shows that one can compute the block LU 
decomposition of type $\partd = (d_1, \ldots, d_r)$ of a matrix $M \in 
M_d(R)$ using a slight modification of Algorithm \ref{algo:stableLU} 
which consists in updating the matrix $L$ (on line \ref{line:updateL}) 
only if $j$ is equal to some $d_1 + \cdots + d_s$ and clearing the 
entries of $L$ below the diagonal of the $s$-th block just after this 
update (\emph{cf} Algorithm \ref{algo:stableBLU}).
\begin{algorithm}
  \SetKwInOut{myInput}{\it Input:}
  \SetKwInOut{myOutput}{\it Output:}
  \SetKwInOut{Void}{}

  \myInput{A partition $\partd = (d_1, \ldots, d_s)$ of a positive
  integer $d$}
  \Void{A matrix $M \in M_d(R)$ known with precision $O(\pi^N)$}
  \myOutput{The matrix $L_{\partd}(M)$}

  \BlankLine

  $\omega$ $\leftarrow$ $M$\;
  $L$ $\leftarrow$ zero matrix of size $d \times d$\;
  $s$ $\leftarrow$ $1$; \: $j_0$ $\leftarrow$ $0$\;
  \For{$j$ from $1$ to $d$}{
    \For{$i$ from $1$ to $j-1$}{
      \lIf{$v_R(\omega_{i,j}) < v_R(\omega_{i,i})$}
        {swap $\omega_j$ and $\omega_i$\;}
      \lIf {$\omega_{i,i} \neq 0$}
        {$s$ $\leftarrow$ $\frac{\omega_{i,j}}{\omega_{i,i}}$ 
         lifted to precision $O(\pi^N)$;\,
         $\omega_j$ $\leftarrow$ $\omega_j - s \cdot \omega_i$\;}
    }
    \If{$j = j_0 + d_s$}{
      \lFor{$j'$ from $j_0+1$ to $j_0 + d_s$}
        {$L_{j'}$ $\leftarrow$ $\frac 1{\omega_{j',j'}} \cdot \omega_{j'}$\;}
      \lFor{$j'$ from $j_0+1$ to $j_0 + d_s$}
        {$L_{j'}$ $\leftarrow$ $L_{j'} - \sum_{i'=j'+1}^{j_0+d_s} L_{i',j'} \cdot L_{i'}$%
        \label{line:clearblock}\;}
      $v$ $\leftarrow$ $\sum_{k=1}^{j_0} v_R(\omega_{k,k})$\;
      \lFor{$j'$ from $j_0+1$ to $j_0 + d_s$ and $i'$ from $j_0+1$ to $d$}
        {set precision of $L_{i',j'}$ to $O(\pi^{N-2v})$\;}
      $j_0$ $\leftarrow$ $j_0 + d_s$; \: $s$ $\leftarrow$ $s+1$\;
    }
  }
  \Return{$L$}\;
\caption{Computing the $L$-part of the LU decomposition of type 
$\partd$\label{algo:stableBLU}}
\end{algorithm}
Furthermore, if the input $M$ is known up to precision $O(\pi^N)$, the
precision of the matrix $L$ returned by Algorithm \ref{algo:stableBLU} is
at least $O(\pi^{N-2\cdot V_{L,\partd}(M)})$. In average, the loss of
precision is then bounded by $2 \cdot \E[V_{L,\partd}] \simeq 2 \cdot
\log_q s$.

All other results proved previously for classical LU decomposition 
(relation with Hermite normal form, notion of L'V' decomposition, 
Hafner-McCauley's improvement) also extend almost \emph{verbatim} to 
block LU decomposition. We will not explain it in details here (but
let the exercise to the reader).

\subsection{Simultaneous PLU decompositions}
\label{subsec:simulPLU}

As we have already said before, a LU decomposition may fail to exist for 
some particular matrices. Nevertheless, it is well known that all 
matrices over a DVR admit a PLU decomposition (recall that a PLU 
decomposition of a matrix $M$ is a factorization $M = PLU$ where $P$ is 
a permutation matrix and $L$ and $U$ are as before) and, in general, 
that several matrices $P$ are possible.

Assume now that we do not pick just one matrix, but a (finite) family of 
matrices $(M_1, \ldots, M_n)$. The question we would like to address is 
the following: does there exist a ``simultaneous PLU decomposition'' of 
the $M_i$'s, that is PLU decomposition of each $M_i$ \emph{with the same 
matrix $P$}. If we want as before $P$ to be a permutation matrix, 
the answer is negative in general. However, if we relax this condition 
and require only that $P$ is invertible (which is enough for certain 
applications, see \S \ref{subsec:sheaf} for a concrete example), the 
answer is positive (at least if the base field is infinite).

The aim of this section is to study this notion of ``simultaneous 
(block) PLU decomposition'' over a base field which is the fraction 
field of a discrete valuation ring.

\subsubsection{The basic result}

Let $(M_1, \ldots, M_n)$ be a family of square $d \times d$ matrices 
over $K$ and fix a partition $\partd = (d_1, \ldots, d_s)$ of $d$. There 
exists an obvious probabilistic algorithm to compute a simultaneous 
block PLU decomposition (of type $\partd$) of the $M_i$'s: we choose $P$ 
at random and compute the block LU decomposition of the $P^{-1} M_i$'s. 
The aim of this paragraph is to prove that this algorithm works quite 
well in the following sense: not only it finds very quickly a matrix $P$ 
that does the job, but it furthermore finds quickly a matrix $P$ for 
which all entries of $P$, $P^{-1}$ and the $L_i$'s are kwown with a good 
precision and do not have a too small valuation. Our precise result can 
be stated as follows.

\begin{theo}
\label{th:simulLU}
Let $n$ be some positive integer. Suppose that for all $m \in \{1, 
\ldots, n\}$ we are given a matrix $M_m \in M_d(K)$ together with a 
finite sequence $\partd_m = (d_{m,1}, \ldots, d_{m,r_m})$ of positive integers whose sum 
equals $d$. Let $\varepsilon$ be a real number in $(0,1)$ and take $v$ 
an integer $\geq \log_q (\frac{r_1 + \cdots + r_n}{q-1}) - \log_q 
\varepsilon$. Then, a random matrix $\omega \in \Omega$ satisfies the 
following conditions with probability at least $1 - \varepsilon$:
\begin{itemize}
\item $\omega$ is invertible in $M_d(K)$ and $\omega^{-1} \in \pi^{-v}
M_d(R)$;
\item for all $m \in \{1, \ldots, n\}$, the matrix $\omega M_m$ admit a 
block LU decomposition of type $\partd_m$ and $L_{\partd_m}(\omega M_m) \in 
\pi^{-v} M_d(R)$;
\end{itemize}
Moreover if the $M_m$'s all lie in $M_d(R)$, are \emph{invertible} in
this ring and are only known with
precision $O(\pi^N)$, one can furthermore require (without changing the 
probability) that, on each input $\omega M_m$, Algorithm 
\ref{algo:stableBLU} outputs $L_{\partd_m}(\omega M_m)$ with precision 
at least $O(\pi^{N-2v})$.
\end{theo}

\begin{rem}
\label{rem:simulLU}
It is also possible to bound the loss of precision if we drop the 
hypothesis of inversibility of the $M_m$'s. The valuations of their
determinants then enter into the scene. (The exercise is left to the
reader.)
\end{rem}

\begin{proof}
Let us begin by proving the first assertion.
Replacing eventually $M_m$ by $M_m + \pi^N$ for a large integer $N$, one 
may assume that all $M_m$'s are invertible in $M_d(K)$. Furthermore, 
since multiplying $M_m$ ($1 \leq m \leq n$) on the right by an upper 
triangular matrix with coefficients in $R$ does not change the matrix 
$L_{\partd_m}(\omega M_m)$ for any $\omega \in M_d(K)$, one can even 
safely assume that all $M_m$'s are invertible in $M_d(R) = \Omega$.
For all $m \in \{1, \ldots, n\}$ and all $s \in \{1, \ldots, r_m\}$, 
define $W_{\partd,s,m} : \Omega \to \N \cup \{\infty\}$ by $W_{\partd_m, 
s,m}(\omega) = V_{\partd,s}(\omega M_m)$ (where $V_{\partd,s}$ is the 
random variable defined in \S \ref{subsec:block}) and set:
$$W = \max (v_R(\det), W_{\partd_1,1,1}, \ldots, W_{\partd_1,1,r_1-1}, 
W_{\partd_2,2,1}, \ldots, W_{\partd_2,2,r_2-1}, \ldots, 
W_{\partd_n,n,1}, \ldots, W_{\partd_n,n,r_n-1}).$$
Since $M_m$ is invertible in $\Omega$, Lemma \ref{lem:Vds} implies that 
$$\Prob[W_{\partd_m,m,s} \leq v] = (1-q^{-v-1}) \cdots (1 - q^{-v-d_s}) 
\leq 1 - \frac {q^{-v}}{q-1}$$
for all $m$, $s$ and $v$. Furthermore by Adbel-Khaffar's Theorem (see 
Theorem 1 of \cite{abdel}) we also know the law of the random variable 
$v_R(\det)$; we have $\Prob[v_R(\det) \leq v] = (1+q^{-v-1}) 
(1+q^{-v-2}) \cdot (1+q^{-v-d})$. Let us simplify this Formula and just 
remember that $\Prob[v_R(\det) \leq v] \geq 1 - (q^{-v-1} + q^{-v-2} + 
\cdots + q_{-v-d}) \geq 1 - \frac{q^{-v}}{q-1}$. We can now estimate the 
law of $W$ as follows:
$$\Prob[W > v] \leq \Prob[v_R(\det) > v] + \sum_{\substack{1 \leq m \leq 
n \\ 1 \leq s < r_m}} \Prob[W_{m,i} > v] \, \leq \, q^{-v} \cdot
\Big( \frac{r_1 + \cdots + r_n - n}{q-1} + \frac 1 {q-1} \Big) \leq 
\varepsilon.$$
Proposition \ref{prop:VLd} shows that $L_{\partd_m}(\omega M_m) \in 
\pi^{-W(\omega)} M_d(R)$ for all $\omega \in \Omega$ and all $m \in \{1, 
\ldots, n\}$ and, on the other hand, it is clear that $\omega^{-1} \in
\pi^{-W(\omega)} M_d(R)$ because $v_R(\det \omega) \leq W(\omega)$. It 
is enough to conclude the proof.

The second assertion (concerning precision) is now clear.
\end{proof}

\begin{rem}
One may wonder if the bound $\log_q (\frac{r_1 + \cdots + r_n}{q-1}) - 
\log_q \varepsilon$ is sharp. Actually, it cannot be for any data of 
$(M_1, \ldots, M_n)$. Indeed, an integer $v$ satisfies the required
conditions of Theorem \ref{th:simulLU} for the families $(M, \ldots, 
M)$ and $(\partd, \ldots, \partd)$ ($n$ times) if and only if it 
satisfies the same conditions for the family reduced to the unique 
matrix $M$. So if $M_1 = \cdots = M_n$ and $\partd_1 = \ldots = \partd_n$, 
one can certainly improve the bound $\log_q (\frac{rn}{q-1}) - 
\log_q \varepsilon$ by removing the facteur $n$ in the first $\log$. 
Nevertheless, by using similar methods as those of \S \ref{sec:stats}, 
one can prove, first, that the result of Theorem \ref{th:simulLU} fails 
if $v \ll \log_q (\frac{\max(r_1, \ldots, r_n)}{q-1}) - \log_q 
\varepsilon$ and, second, that if $M_1, \ldots, M_n$ are themselves 
chosen randomly, it even fails for $v \ll \log_q (\frac{r_1 + \cdots + 
r_n}{q-1}) - \log_q \varepsilon$ (\emph{i.e.} the given bound is sharp).
\end{rem}

\subsubsection{A slight generalization}

For the application we have in mind (see \S \ref{subsec:sheaf}), we will 
need a slight generalization of Theorem \ref{th:simulLU} where the 
matrices $M_1, \ldots, M_n$ on the one hand and the matrix $\omega$ on 
the other hand are not defined over the same field. Let $\tilde K$ be a 
finite extension of $K$. A classical result asserts that the valuation 
$v_R$ extends uniquely to $\tilde K$. Let $R$ be the ring of integers of 
$\tilde K$, that is the subset of $\tilde K$ consisting of elements with 
nonnegative valuation. Set finally:
$$\Pi(q) = q\cdot \prod_{i=1}^\infty \: (1-q^{-i}).$$
It is easy to check that $q-1 - \frac 1{q-1} < \Pi(q) < q-1$.

\begin{theo}
\label{th:simulLU2}
Let $n$ be some positive integer. Suppose that for all $m \in \{1, 
\ldots, n\}$ we are given a matrix $M_m \in M_d(\tilde K)$ together with 
a finite sequence $\partd_m = (d_{m,1}, \ldots, d_{m,r_m})$
of positive integers whose 
sum equals $d$. Let $\varepsilon$ be a real number in $(0,1)$ and take 
$v$ an integer $\geq \log_q (\frac{r_1 + \cdots + r_n}{\Pi(q)}) - \log_q 
\varepsilon$. Then, a random matrix $\omega \in \Omega$ satisfies the 
following conditions with probability at least $1 - \varepsilon$:
\begin{itemize}
\item $\omega$ is invertible in $M_d(K)$ and $\omega^{-1} \in \pi^{-v}
M_d(R)$;
\item for all $m \in \{1, \ldots, n\}$, the matrix $\omega M_m$ admit a 
block LU decomposition of type $\partd$ and $L_{\partd_m}(\omega M_m) \in 
\pi^{-v} M_d(\tilde R)$;
\end{itemize}
Moreover if the $M_m$'s all lie in $M_d(R)$, are \emph{invertible} in 
this ring and are only known with 
precision $O(\pi^N)$, one can furthermore require (without changing the 
probability) that, on each input $\omega M_m$, Algorithm 
\ref{algo:stableBLU} outputs $L_{\partd_m}(\omega M_m)$ with precision at 
least $O(\pi^{N-2v})$.
\end{theo}

\begin{rem}
\label{rem:simulLU2}
Once again (see Remark \ref{rem:simulLU}), one can bound the loss of 
precision as well if we drop the hypothesis of inversibility of the 
$M_m$'s and put into the machine the valuations of all $\det M_m$.
\end{rem}

We now start the proof of Theorem \ref{th:simulLU2}; it will occupy the 
rest of this subsection. As in the proof of Theorem \ref{th:simulLU}, we 
start with the first assertion and assume that $M_m$ is invertible in 
$M_d(\tilde R)$ for all $m$. However, in our new settings, this fact no 
longer implies that $\omega M_m$ runs over $M_d(R)$ when $\omega$ runs 
over $M_d(R)$. Thus, we can no longer work with the random variables 
$\omega \mapsto V_{\partd_m,i} (\omega M_m)$ and we need to modify a bit our 
strategy. Actually, since we just want to bound from above --- and not 
from below --- the valuation of the matrices $L_{\partd_m}(\omega M_m)$, we 
can argue first assuming that $m$ is fixed and then add probabilities. 
Moreover, by the proof of Proposition \ref{prop:VLd} (see also Formula 
\eqref{eq:CramerL} when $\partd_m = (1, \ldots, 1)$), bounding the valuation 
of $L_{\partd_m}(\omega M_m)$ reduces to bounding the valuation of the 
$(d_{m,1} + \ldots + d_{m,s})$-th minor of $\omega M_m$ for all $s \in \{1, 
\ldots, r_m-1\}$. Thus we first fix $m \in \{1, \ldots, n\}$ and $s \in 
\{1, \ldots, r_m-1\}$ and look for an upper bound for the valuation of the 
$j_m(s)$-th minor of $\omega M_m$ where, by definition, $j_m(s) = d_{m,1} + 
\cdots + d_{m,s}$.
For $1 \leq i \leq j_m(s)$, we are going to define a random variable $W_i : \Omega 
\to \frac 1 e \N \cup \{\infty \}$ where $e$ is the ramification index
of $\tilde K / K$ (\emph{i.e.} the index of $v_R(K^\star)$ is $v_R
(\tilde K^\star)$). The construction of the $W_i$'s is achieved by 
applying the classical algorithm of LU decomposition. We pick $\omega 
\in \Omega$ and first set $M^{(1)} = (\omega M_m)_{\{1, \ldots, r_m\}, 
\{1, \ldots, r_m\}} \in M_r(\tilde R)$. Let $j$ be the first index for 
which $v_R(M^{(1)}_{1,j})$ is minimum among the valuations of all 
entries of the first line of $M^{(1)}$. Let $M^{(2)}$ be the matrix
obtained from $M^{(1)}$ by swapping the $j$-th column with the first one 
and by clearing all the entries of the first row (expect the first one) 
by pivoting, \emph{i.e.} adding to each column (expect the first one) a 
suitable multiple of the first one. The matrix $M^{(2)}$ looks like
$$\left( \begin{matrix}
\star & 0 & \cdots & 0 \\
\star & \cdots & \cdots & \star \\
\vdots & & & \vdots \\
\star & \cdots & \cdots & \star
\end{matrix} \right).$$
We now continue this process: we select the first index $j \geq 2$ for 
which $v_R(M^{(2)}_{2,j})$ is minimal, we obtain $M^{(3)}$ by putting 
the $j$-th column in the second position and clearing all the other 
entries on the second row. Repeating this again and again, we obtain
a finite sequence $M^{(1)}, \ldots, M^{(j_m(s))}$ of matrices and the 
last one is lower triangular.
For $i \in \{1, \ldots, j_m(s)\}$, we define $W_i(\omega)$ as the 
valuation of the $i$-th diagonal entry of $M^{(j_m(s))}$ (or equivalently 
of $M^{(j)}$ for some $j \geq i$). It is clear that the determinant of 
the $j_m(s)$-th principal minor of $\omega M_m$ is equal to $W_1(\omega) + 
W_2(\omega) + \cdots + W_{j_m(s)}(\omega)$. We need to determine the law 
and the correlations between the $W_i$'s. We begin by a Lemma.

\begin{lemma}
\label{lem:probf}
Let $f : \tilde R^d \to \tilde R^{r}$ be a \emph{surjective} map. Then
$$\Prob \big[ \, f(x) \in \pi^v \tilde R^{r} \,|\, x \in R^d \, \big] 
\, \leq \, q^{-rv}$$
for all nonnegative integer $v$.
\end{lemma}

\begin{proof}
Let $\tilde k$ denote the residue field of $\tilde R$; it is a finite 
extension of $k$. Since $f$ is surjective, it induces a surjective 
$\tilde k$-linear map $\bar f : \tilde k^d \to \tilde k^{r}$ over the 
residue field. Moreover, the image of $\bar f$ is generated over $\tilde 
k$ by $\bar f(k^d)$. Thus $\dim_k \bar f(k^d) = \dim_{\tilde k} \bar f 
(\tilde k^d) = r$. This fact implies the existence of a $k$-linear map 
$\bar g : \tilde k^{r} \to k^{r}$ such that the composite $\bar g \circ \bar 
f : k^d \to k^{r}$ is surjective. Let $g : \tilde R^{r} \to R^{r}$ be any 
$R$-linear lifting of $\bar g$. The $R$-linear morphism $h = g \circ f : 
R^d \to R^{r}$ induces a surjection over the residue field and thus is 
itself surjective. Furthermore, it is clear that $h(x)$ is divisible by 
$\pi^v$ if $f(x)$ is. Hence, for $x$ staying in $R^d$, we have $\Prob [ 
f(x) \in \pi^v \tilde R^{r} ] \leq \Prob [ h(x) \in \pi^v R^{r} ]$ and we 
are reduced to prove the Lemma with $f$ replaced by $h$. (In other words, 
we may assume that $\tilde K = K$.)

By the structure Theorem for finitely generated modules over a principal 
domain (recall that $R$ is a principal domain), there exists a basis 
$(e_1, \ldots, e_d)$ of $R^d$ such that the first $(d-r)$ vectors $e_1, 
\ldots, e_{d-r}$ form a basis of $\ker h$. Now, using that $h$ is 
surjective, we easily see that a vector $x = \sum_{i=0}^d x_i e_i \in 
R^d$ satisfies $h(x) \in \pi^v R^{r}$ if and only if $x_i$ is divisible by 
$\pi^v$ for all $i > d-r$. But, the probability that such an event 
occurs is $q^{-rv}$ and we are done.
\end{proof}

\begin{cor}
\label{cor:lawWi}
For all integers $v_1, \ldots, v_{j_m(s)}$, we have:
$$\Prob [ W_i \geq v_i, \, \forall i ] \leq \prod_{i=1}^{j_m(s)}
q^{-(r+1-i)v_i}.$$
\end{cor}

\begin{proof}
For $\omega \in \Omega$ and $i \in \{1, \ldots, j_m(s)\}$, we denote by 
$\omega_i$ the $i$-th row of $\omega$ (and consider it as a vector of 
$R^d$) and by $M^{(i)}(\omega)$ the matrix defined above.
Let $F^{(1)}$ be the submatrix of $M_m$ consisting of its first $j_m(s)$ 
columns and let $f^{(1)} : \tilde R^d \to \tilde R^{j_m(s)}$ be the $\tilde 
R$-linear map whose matrix is $\t F^{(1)}$. The fact that $M_m$ is 
invertible implies that $f^{(1)}$ is surjective. Lemma \ref{lem:probf} 
applied to $f^{(1)}$ yields:
\begin{equation}
\label{eq:probW1}
\Prob[W_1 \geq v_1] \leq q^{-j_m(s)v_1}.
\end{equation}
Now remember that $M^{(2)}(\omega)$ is obtained from $M^{(1)}(\omega)$
by performing a sequence of elementary operations on columns. It then
exists a matrix $P^{(1)}(\omega)$ such that $M^{(2)}(\omega) = M^{(1)}
(\omega) \cdot P^{(1)}(\omega)$. Clearly $P^{(1)}(\omega)$ depends 
only on $\omega_1$ and we will denote it $P^{(1)}(\omega_1)$ in the
sequel. Set $F^{(2)}(\omega_1) = F^{(1)} \cdot P^{(1)}(\omega_1)$
and let $f^{(2)}(\omega_1) : \tilde R^d \to \tilde R^{j_m(s)}$ denote the map whose 
matrix is $\t F^{(2)}(\omega_1)$. It is surjective and one can then apply Lemma
\ref{lem:probf} to the composite $\pr_{j_m(s)-1} \circ f^{(2)}(\omega_1)$ where 
$\pr_{j_m(s)-1} : \tilde R^{j_m(s)} \to \tilde R^{j_m(s)-1}$ is the projection on the 
first coordinates. It gives $\Prob[W_2 \geq v_2 \:|\: \omega_1 = x_1 ] 
\leq q^{-(j_m(s)-1)v_2}$ for all $x_1 \in R^d$. Integrating now over $x_1$ 
and using \eqref{eq:probW1}, we get
$$\Prob[W_1 \geq v_1 \text{ and } W_2 \geq v_2 ] \leq q^{-j_m(s)v_1} \cdot 
q^{-(j_m(s)-1)v_2}.$$
The Corollary follows by repeating $j_m(s)$ times the previous argument.
\end{proof}

If we denote by $\delta_{s,m}(\omega)$ the valuation of the determinant 
of the $j_m(s)$-th principal minor of $\omega M_m$, Corollary 
\ref{cor:lawWi} allows us to do the following computation:
\begin{eqnarray*}
\Prob[\delta_{s,m} > v] 
& \leq & \sum_{\substack{v_1, \ldots, v_{j_m(s)} \geq 0 \\ v_1 + \cdots + v_{j_m(s)} = v+1}} \Prob[W_i \geq v_i, \, \forall i] 
\,\, \leq \sum_{\substack{v_1, \ldots, v_{j_m(s)} \geq 0 \\ v_1 + \cdots + v_{j_m(s)} = v+1}} q^{-(v_1 + 2 v_2 + \cdots + j_m(s) v_{j_m(s)})}  \\
& = & q^{-v-1} \sum_{\substack{v_2, \ldots, v_{j_m(s)} \geq 0 \\ v_2 + \cdots + v_{j_m(s)} \leq v+1}} q^{-(v_2 + 2 v_3 + \cdots + (j_m(s)-1) v_{j_m(s)})} \\
& \leq & q^{-v-1} \sum_{v_2, \ldots, v_{j_m(s)} \geq 0} q^{-(v_2 + 2 v_3 + \cdots + (j_m(s)-1) v_{j_m(s)})}  \\
& = & q^{-v-1} \cdot \Bigg(\sum_{v_2=0}^\infty q^{-v_2}\Bigg) \cdot
\Bigg(\sum_{v_3=0}^\infty q^{-2v_3}\Bigg) \cdots
\Bigg(\sum_{v_{j_m(s)}=0}^\infty q^{-(j_m(s)-1)v_{j_m(s)}}\Bigg) \\
& = & q^{-v-1} \cdot (1-q^{-1})^{-1} (1-q^{-2})^{-1} \cdots (1-q^{-j_m(s)+1})^{-1}
\,\, \leq \,\, q^{-v} \cdot \Pi(q)^{-1}.
\end{eqnarray*}
It is time now to free $s$ and $m$: summing the above estimation over 
all possible $s$ and $m$, we find that $\delta(\omega) = \max_{s,m} 
\delta_{s,m}(\omega)$ is greater than $v$ --- which implies that 
$L_{\partd}(\omega M_m)$ does not lie in $\pi^{-v} M_d(R')$ --- with probability 
at most $(r_1 + \cdots + r_n - n) \cdot q^{-v} \cdot \Pi(q)^{-1}$ and consequently that 
$\omega$ does not satisfy the conditions of the first statement of 
Theorem \ref{th:simulLU2} with probability at most:
$$q^{-v} \cdot \Bigg( \frac 1{q-1} + \frac {r_1 + \cdots + r_n - n} {\Pi(q)}
\Bigg) \leq q^{-v} \cdot \frac{r_1 + \cdots + r_n}{\Pi(q)}.$$
Hence if $v$ is chosen $\geq \log_q(\frac{r_1 + \cdots + r_n}{\Pi(q)}) - \log_q
(\varepsilon)$, this probability is less than $\varepsilon$: the first
part of Theorem \ref{th:simulLU2} is proved.

The second part now follows easily: indeed, we know that, on the input 
$\omega M_m$, the Algorithm \ref{algo:stableLU} decreases the precision 
by a factor that cannot exceed $\pi^{2 \max_i V_i(\omega M_m)}$ and so,
\emph{a fortiori}, by a factor that cannot exceed $\pi^{2 \delta(M)}$
where $\delta$ is the random variable defined above. The conclusion
follows from this.

\begin{rem}
The bound of Theorem \ref{th:simulLU2} is sharp if $M_1, \ldots, M_n$ 
are chosen randomly among all square $d \times d$ matrices with 
coefficients in $R$. However, it is not true in general and it is even 
not true if 
$M_1, \ldots, M_n$ are chosen randomly among all matrices over $\tilde 
R$. Indeed, in that case, using results of \S \ref{sec:stats}, one can
prove that, in average, the better possible bound for $v$ is given by:
$$[\tilde K : K]^{-1} \cdot {\textstyle \Big( \log_q(\frac {r_1 + 
\ldots + r_n}
{\Pi(q)}) - \log_q (\varepsilon)} \Big) + O(1)$$
with an extra factor $[\tilde K : K]^{-1}$, which can be very small.
\end{rem}

\subsection{Modules over $K[X]$ and sheaves over $\A^1_K$}
\label{subsec:sheaf}


Let $X$ denote an affine curve over $K$ and $A = K[X]$ be the ring of 
regular functions over $X$. It is well known that the category of 
coherent sheaves over $X$ is equivalent to that of finitely generated 
modules over $A$. In particular, the data of a submodule $M \subset A^d$ 
(for some fixed integer $d$) is equivalent to the data of a coherent 
subsheaf $\calM \subset \O_X^d$. Nevertheless, these two objects are of 
different nature and we would like to represent them in two different
ways:
\begin{itemize}
\item a submodule $M \subset A^d$ by a matrix of generators
\item a subsheaf $\calM \subset \O_X^d$ by the data of the stalk 
$\calM_x \subset \O_{X,x}^d$ for each closed point $x \in X$ (note that 
this inclusion is not trivial for only a finite number of points $x$).
\end{itemize}
Since these objects are supposed to be equivalent, it is natural to ask 
if one can find an efficient way to go from one representation to the 
other. Actually going from the global description to the local one is 
quite easy: it suffices to localize at each point $x$. Contrariwise,
going in the opposite direction is not so obvious and will be discuss
now.

\medskip

\emph{From now on, we assume for simplicity that $X$ is the affine line 
$\A^1_K$} (and leave to the reader the exercise to extend our 
constructions to a more general setting). With this extra assumption, 
the ring $A$ is nothing but the ring of univariate polynomials with 
coefficients in $K$.

\subsubsection{Rephrasing our problem in concrete terms}

For all irreducible polynomials $P \in K[X]$, let $A_P$ denote the 
completion of $A$ for the $P$-adic topology, that is $A_P = 
\varprojlim_r A/P^r A$. Concretely $A_P$ can be identified with a ring 
of power series with coefficients in the residue field $K_P = A/ P A$ in 
one indeterminate $X_P$. This variable $X_P$ should be thought as 
``$X-a_P$'' where $a_P$ is a (fixed) root of $P$ in $K_P$. Under the 
identification $A_P \simeq K_P[[X_P]]$, the natural embedding $A \to 
A_P$ is just the Taylor expansion at $a_P$:
$$F(X) \mapsto \sum_{i=0}^\infty \frac{F^{(i)}(a_P)}{i!} \cdot X_P^i.$$
Let $P_1, \ldots, P_n$ be the minimal polynomials of $a_1 \ldots, a_n$ 
respectively and, for simplicity, set $A_m = A_{P_m}$. The question we 
have addressed earlier is then equivalent to the following: given, for 
all $m \in \{1, \ldots, n\}$, a submodule $\calM_m \subset A_m^d$ free 
of maximal rank, how can one find explicitely a $A$-module $\calM 
\subset A^d$ such that $A_m \otimes \calM = \calM_m$ (as a submodule of 
$A_m^d$) for all $m$ and $A_P \otimes \calM = A_P^d$ for all other $P$?

One can actually rephrase again this question in very concrete terms by 
taking basis everywhere. Indeed, if $B$ is $A$ or one of the $A_m$'s, 
any free submodule of $B^d$ of rank $d$ can certainly be represented by 
a square $d \times d$ matrix with coefficients in $B$: the module is 
recovered from the matrix by taking the span of its column vectors. Note 
furthermore that two matrices $G$ and $H$ defines the same module if and 
only if there exists an invertible matrix $P$ with coefficients in $B$ 
such that $G = HP$; if this property holds, we shall say that $G$ and 
$H$ are \emph{right-equivalent}. Since all our base rings are principal 
domains, we know that any matrix $G \in M_d(B)$ admits a factorization 
$G = M D N$ where $M$ and $N$ are two invertible matrices, $D$ is 
diagonal and each diagonal entry of $D$ divides the next one. Up to 
replacing $G$ by a right-equivalent matrix, one can furthermore assume 
that $N$ is the identity matrix, \emph{i.e.} that $G$ has the particular 
form $G = M D$. Moreover, if $B$ is one of the $A_m$'s, it is safe to 
assume that the diagonal entries of $D$ are all some powers of the 
variable $X_m$ since all nonvanishing element of $A_m$ can be written as 
a product of an invertible element with a power of $X_m$. In that case 
the data of $D$ is then reduced to that of a nondecreasing sequence of 
integers $n_1 \leq \ldots \leq n_d$.

\medskip

With all these remarks, our question becomes:

\begin{problem}
\label{problem:sheaf}
Given for all $m$, an invertible matrix $M_m \in M_d(A_m)$ and a
nondecreasing sequence of $d$ integers $e_{m,1} \leq \cdots \leq
e_{m,d}$, how can one construct explicitely a couple $(M,D)$ of 
matrices over $A$ such that:
\begin{itemize}
\item[i)] the matrix $M$ is invertible in $M_d(A)$ ;
\item[ii)] the matrix $D$ is diagonal and each of its diagonal entry 
divides 
the next one ;
\item[iii)] for all $m \in \{1, \ldots, n\}$, the matrix $MD$ is 
right-equivalent to $M_m D_m$ over $A_m$ where $D_m = 
\Diag(X_m^{e_{m,1}}, \ldots, X_m^{e_{m,d}})$ ;
\item[iv)] for all irreducible polynomial $P \in K[X]$ which is not one 
of the $P_m$'s, the matrix $MD$ is right-equivalent to the identity 
matrix over $A_P$.
\end{itemize}
\end{problem}

\subsubsection{The answer}
\label{subsec:answer}

We consider, for all $m \in \{1, \ldots, n\}$, an invertible matrix 
$M_m \in M_d(A_m)$ together with a nondecreasing sequence of $d$ 
integers $e_{m,1} \leq \cdots \leq e_{m,d}$. Our aim is to construct
a couple $(M,D)$ satisfying the Conditions i), ii), iii) and iv)
above. Firstable, we define the matrix $D$ as follows:
$$D = \Diag(P_1^{e_{1,1}} \cdots P_n^{e_{n,1}}, \ldots, \:
P_1^{e_{1,d}} \cdots P_n^{e_{n,d}}).$$
It clearly satisfies Condition ii).

\begin{lemma}
\label{lem:condiiv}
Let $M \in M_d(A)$. 

a) Assume that, for all $m \in \{1, \ldots, n\}$, the 
matrix $M$ considered as an element of $M_d(A_m)$ (\emph{via} the 
natural embedding $A \to A_m$) is congruent to $M_m$ modulo 
$X_m^{e_{m,d}+1}$. Then, the couple $(M,D)$ satisfies Condition 
iii).

b) Assume moreover that $M$ is invertible in $M_d(A)$. Then the 
couple $(M,D)$ satisfies Conditions i), ii), iii) and iv).
\end{lemma}

\begin{proof}
Note that in the ring $A_m$, the polynomial $P_m$ is equal to the
product of $X_m$ by a unit whereas all other $P_{m'}$'s (for $m' \neq
m)$ are invertible. We deduce from this that $D$ is right-equivalent to 
$D_m$ over $M_m$. Hence, our first hypothesis implies that $MD$ is 
right-equivalent to a matrix congruent to $M_m D_m$ modulo 
$X_m^{e_{m,d}+1}$. In other words, there exists a matrix $Q \in \GL_d 
(M_m)$ such that $M D$ is right equivalent to:
$$M_m D_m + X_m^{e_{m,d}+1} Q = M_m D_m \cdot \big[I_d + X_m \cdot
\Diag(X_m^{e_{m,d}-e_{m,1}}, \ldots, \:X_m^{e_{m,d}-e_{m,d-1}}, \:1) 
\cdot Q \big]$$
where $I_d$ is of course the identity matrix. The last factor (the one 
between brackets) is a matrix over $A_m$ congruent to identity modulo 
$X_m$. It is therefore invertible. 
It follows that $M D$ is right-equivalent to $M_m D_m$, and part a)
of the Lemma is proved.

We assume now that $M$ is invertible. Then, clearly, Condition i) holds. 
Moreover, we have already seen that Conditions ii) are iii) are 
fulfilled. It is then enough to prove Condition iv). Let $P \in K[X]$ be 
an irreducible polynomial different from all the $P_m$'s. All $P_m$'s 
are then invertible in $A_P$ and, consequently, so is the matrix $D$. 
Since $M$ is itself invertible, the product $MD$ belongs to $\GL_d(A_P)$ 
and is then right-equivalent to the identity matrix.
\end{proof}

It is actually not difficult to produce a matrix $M$ satisfying the 
assumption of the Lemma \ref{lem:condiiv}.a). Indeed, the identification 
$A_m/X_m^{e_{m,d}+1}A_m \simeq A / P_m^{e_{m,d}+1} A$ shows that the 
congruence $M \equiv M_m \pmod {X_m^{e_{m,d}+1}}$ is equivalent to $M 
\equiv M'_m \pmod {P_m^{e_{m,d}+1}}$ for a certain matrix $M'_m \in 
M_d(A)$. Hence, finding a convenient $M$ is just a direct application of 
the Chinese Remainder Theorem (recall that all $P_m$'s are irreducible 
and pairwise distinct polynomials).

Producing a matrix $M$ satisfying also the second assumption of Lemma 
\ref{lem:condiiv} is a bit more tricky but can be achieved using block 
LU decomposition. For $m \in \{1, \ldots, n\}$, let $r_m$ be the numbers 
of differents values taken by the sequence $(e_{m,1}, \ldots, e_{m,d})$ 
and let $d_{m,s}$ ($1 \leq s \leq r_m$) denote the number of times this 
sequence takes its $i$-th smallest value. We then have:
$$e_{m,1} = \cdots = e_{m,d_{m,1}} < e_{m,d_{m,1} + 1} = e_{m,d_{m,1} + 2} =
\cdots < e_{m,d_{m,1} + d_{m,2}} < e_{m, d_{m,1}+d_{m,2}+1} = \cdots$$
Assume now for a moment that all $M_m$'s admit a block LU 
factorization $M_m = L_m U_m$ of type $\partd_m = (d_{m,1}, \ldots, 
d_{m,r_m})$. Since $M_m$ is invertible, so is $U_m$. 
Let $V_m$ be the matrix obtained from $U_m$ by multiplying its 
$(i,j)$-th entry by $X_m^{e_{m,j}-e_{m,i}}$ (note that the exponent is 
always nonnegative when the $(i,j)$-th entry of $U_m$ does not vanish). 
Obviously, $V_m$ is again upper triangular and its diagonal entries are 
equal to those of $U_m$. Thus $U_m$ and $V_m$ share the same determinant 
and $V_m$ is invertible. Moreover, we check that $M_m D_m = L_m D_m 
V_m$, from what we derive that $M_m D_m$ is right-equivalent to $L_m 
D_m$. Since all $L_m$'s are unit lower triangular, there certainly 
exists a unit lower triangular matrix $L \in M_d(A)$ which is congruent 
to $L_m$ modulo $X_m^{e_{m,d}+1}$ for all $m$. Such a matrix is 
apparently invertible and also satisfies the assumption in part a) of 
Lemma \ref{lem:condiiv}. We can then simply take $M = L$.

Now let us go back to the general case where some $M_m$ might not have a 
block LU decomposition of type $\partd_m$. In that case, we denote by 
$M_m(0) \in M_d(K_m)$ the image of $M_m$ under the canonical projection 
$A_m \to A_m/P_m A_m = K_m$ (or, equivalently, $A_m \simeq K_m[[X_m]] 
\to K_m$). The coefficients of $M_m(0)$ then all lie in $K_m$, which is 
a finite extension of $K$. We can therefore apply Theorem 
\ref{th:simulLU2} which implies in particular the existence of a matrix 
$\omega \in \GL_d(K)$ such that $\omega \cdot M_m (0)$ has a block LU 
decomposition of type $\partd_m$ for all $m$. Lemma \ref{lem:liftLU} 
below shows that this decomposition lifts to a LU decomposition of type 
$\partd_m$ of $\omega M_m$.

\begin{lemma}
\label{lem:liftLU}
Let $L$ be a finite extension of $K$. Pick $M \in M_d(L[[Y]])$ and
denote by $M(0)$ its image in $M_d(L)$ under the projection $L[[Y]]
\to L$, $Y \mapsto 0$. Assume that $M(0)$ is invertible and admits 
a block LU decomposition of type $\partd$ for a certain partition 
$\partd$ of $d$. Then $M$ also does.
\end{lemma}

\begin{proof}
Write $\partd = (d_1, \ldots, d_r)$ and set as usual $j(s) = d_1 + \cdots + 
d_s$ for all $s \in \{1, \ldots, r\}$. It is enough to check that, for 
all $s$, the $j(s)$-th principal minor of $M$, say $\delta_s(M)$, is 
invertible in $L[[Y]]$. But, $L[[Y]]$ being a local ring, $\delta_s(M)$ 
is invertible if and only if its image $\delta_s(M)(0)$ is. Now 
remark that this image is nothing but the corresponding minor of $M(0)$: 
in other words $\delta_i(M)(0) = \delta_i(M(0))$. The invertibility of 
$M(0)$ together with the fact that it has a block LU decomposition of 
type $\partd$ shows that $\delta_i(M(0))$ is invertible in $L$ and we 
are done.
\end{proof}

\noindent
We are now in position to argue as above. For all $m$, write $\omega M_m 
= L_m U_m$ the block LU decomposition of $M_m$ of type $\partd$. By the 
Chinese Remainder Theorem, there exists a unit lower triangular matrix 
$L$ with coefficients in $A$ such that $L \equiv L_m 
\pmod{X_m^{e_{m,d}+1}}$ for all $m$. The matrix $M = \omega^{-1} L$ then 
satisfies the two assumptions of Lemma \ref{lem:condiiv}. Hence, it 
satisfies also the conclusions of this Lemma and we have solved our 
problem. Algorithm \ref{algo:sheaf} summarizes the different steps of 
the proposed solution.
\begin{algorithm}
  $D$ $\leftarrow$ 
  $\Diag(P_1^{e_{1,1}} \cdots P_n^{e_{n,1}}, \ldots, \:
  P_1^{e_{1,d}} \cdots P_n^{e_{n,d}})$\;
  $\omega$ $\leftarrow$ a random matrix in $M_d(R)$\;
  \For {m from 1 to n}{
    compute $\partd_m = (d_{m,1}, \ldots, d_{m,r_m})$\;
    {$L_m$ $\leftarrow$ $L_{\partd_m} (\omega M_m)$ (computed by Algorithm \ref{algo:stableBLU})\label{line:Lm}\;}
  }
  $L$ $\leftarrow$ a unit lower triangular matrix in $M_d(A)$ such that
  $L \equiv L_m \pmod{X_m^{e_{m,d}+1}}$ for all $m$\label{line:chinese}\;
  \Return $(\omega^{-1} L, D)$\;
\caption{A solution to Problem \ref{problem:sheaf}\label{algo:sheaf}}
\end{algorithm}
Of course, if $\omega$ is not invertible or one of the $\omega M_m$'s 
does not admit a block LU decomposition of the required type, Algorithm
\ref{algo:sheaf} fails. If it happens, we simply rerun the algorithm
again and again until it works: it follows from Theorem \ref{th:simulLU2} 
that we will get the desired answer quite fast.

\medskip

Let us analyze quickly how much precision is loss in average by this 
method. In order to fix ideas, let us assume that the entries of the 
matrix $M_m$ are explicitely given as polynomials in $K_m[X_m]$ 
(eventually modulo $X_m^{e_{m,d}+1}$) and that all these polynomials are 
known with precision $O(\pi^N)$ for some integer $N$. For simplicity, we 
assume moreover that $M_m(0)$ has coefficients in the ring of integers 
$R_m$ of $K_m$ and that it is invertible in 
$M_d(R_m)$\footnote{Otherwise, we would need to take in account the 
valuation of $\det M_m(0)$ as in Remarks \ref{rem:simulLU} and 
\ref{rem:simulLU2}.}. Set as before $j_m(s) = d_{m,1} + \ldots + 
d_{m,s}$ and, for all admissible pair $(m,s)$, let denote by $D_{m,s}$ 
the determinant of the $j_m(s)$-th principal minor of $\omega M_m$. 
Define also $\delta(\omega)$ to the maximum of all $v_R(D_{m,s}(0))$ 
when $m$ and $s$ run over all the possibilities. By the proof of Theorem 
\ref{th:simulLU2}, we know that $\delta(\omega)$ is less than 
$$\textstyle v = \log_q (\frac 2{\Pi(q)}) + \log_q (r_1 + \cdots + 
r_n)$$
with probability at least $\frac 1 2$. In many concrete situations, it 
is not easy to compute exactly the $r_m$'s but it will nevertheless in 
general quite simple to estimate them. Indeed, going back to the 
definition, it is clear that $r_m$ is less than both $d$ and $e_{m,d}$ 
and these latter quantites are natural parameters on which we will in 
general have a good control (\emph{cf} \cite{caruso-lubicz}, \S 3.2 for 
a concrete example).
From now on, we assume that all matrices $M_m$ computed on line 
\ref{line:Lm} satisfy this estimation. If this property does not hold, 
we simply agree to rerun Algorithm \ref{algo:sheaf} until the desired 
property holds.

The next step is to measure the size of the denominators appearing in 
the following nonconstant coefficients. In order to do this, we 
introduce a new parameter $w$ by requiring that all matrices $M_m$ have 
coefficients in the ring $R_{m,w}$ defined as the image of $R_m 
[\frac{X_m}{\pi^w}]$ in the quotient ring $K_m[X_m]/X_m^{e_{m,d}+1}$. 
Clearly $D_{m,s}$ belongs to $R_{m,w}$ for all $(m,s)$ and, by we have 
said before, it has a representant whose constant coefficient has a 
valuation less than $v$. Its inverse $D_{m,s}^{-1}$ then belongs to 
$\pi^{-v} \cdot R_{m,v+w}$ and is known up to an element of $\pi^{N-2v} 
\cdot R_{m,v+w}$. All entries of $L_m$ will consequently be known with 
this precision.

It remains to analyze the line \ref{line:chinese} of Algorithm 
\ref{algo:sheaf}. Note that the matrix $L$ we want to compute can be 
expressed in terms of the $L_m$'s by the formula $L = C_1 L'_1 + \cdots 
+ C_n L'_n$ where:
\begin{itemize}
\item $L'_m$ is a matrix with coefficients in $K[X]$ whose reduction
modulo $P_m^{e_{d,m}+1}$ corresponds to $L_m$ \emph{via} the natural
isomorphism:
\begin{equation}
\label{eq:isomF}
K[X]/P_m^{e_{d,m}+1} \stackrel{\sim}{\longrightarrow} 
K_m[X_m]/X_m^{e_{d,m}+1}, \quad F(X) \mapsto \sum_{i=0}^{e_{d,m}} 
\frac{F^{(i)}(a_m)}{i!} \cdot X_m^i
\end{equation}
\item $C_m$ is a polynomial congruent to $1$ modulo $P_m^{e_{d,m}+1}$
and divisible by $P_{m'}^{e_{d,m'}+1}$ for all $m' \neq m$.
\end{itemize}
In order to bound the loss of precision as we would like to do, we 
assume for simplicity that all $P_m$'s are entirely known. We introduce 
again two new parameters. The first one is an integer $v_1$ for which we
require that the image of $R[X]/P_m^{e_{d,m}+1}$ under the isomorphism 
\eqref{eq:isomF} contains $\pi^{v_1} \cdots R_m[X_m]/X_m^{e_{d,m}+1}$ 
for all $m$. The second parameter is the integer $v_2$ defined as the 
opposite of the smallest valuation of a coefficient of the unique 
polynomial $C_m$ of degree $< \sum_{m=1}^n (e_{d,m}+1) \deg P_m$ 
satisfying the above condition. Now, remember that we have proved that 
$L_m$ are known up to an element of $\pi^{N-2v} \cdot R_{m,v+w}$. It is
then \emph{a fortiori} known up to an element of $\pi^{N-2v-e(v+w)}$ 
where $e = \max(e_{1,d}, \ldots, e_{n,d})$. Inverting the isomorphism 
\eqref{eq:isomF}, we find that $L'_m$ is certainly known modulo 
$\pi^{N-2v-e_{m,d}(v+w)-v_1}$. Finally the formula $L = C_1 L'_1 + 
\ldots + C_n L'_n$ shows that $L$ is known with precision $O(\pi^{N - 
2v-e(v+w)-v_1-v_2})$ (recall that we have assumed that the $P_m$'s --- 
and consequently the $C_m$'s --- are known with infinite precision). The 
total loss of precision of Algorithm \ref{algo:sheaf} is then bounded by 
$2v + e(v+w) + v_1 + v_2$.

\begin{rem}
The parameters $v_1$ and $v_2$ are \emph{not} easy to estimate in 
general. One can nevertheless keep in mind the following: $v_1$ measures 
the ramification of the roots $a_m$ of the $P_m$'s and $v_2$ measures 
the distance between these roots. For instance, to be more precise, one 
can easily prove that if all $a_m$ lie in the ring of integers of an 
unramified extension $K'$ of $K$ then one can just take $v_1 = 0$. If in 
addition the $a_m$'s are pairwise distinct in the residue field of $K'$ 
(which is the same as to be distinct in the residue field of $K$ since $K'/K$ is 
unramified), one can also take $v_2 = 0$. In that very particular case, 
the computation of line \ref{line:chinese} does not generate any loss of 
precision. We refer to \cite{caruso-lubicz} for a quite different 
example where the constants $v_1$ and $v_2$ do not vanish but stay 
nevertheless under control.
\end{rem}

Here is a final important remark. Algorithm \ref{algo:sheaf} still works 
if, instead of computing (the $L$-part) of the block LU decomposition of 
$\omega M_m$, we compute a unit lower triangular (and not \emph{block} 
unit lower triangular) $L_m$ such that there exists a block upper 
triangular (with respect to $\partd$) matrix $U_m$ with the property 
that $M_m = L_m U_m$. Indeed, the knowledge of these $L_m$'s is enough 
to compute $L$ (which need to be only unit lower triangular) and then to 
conclude using Lemma \ref{lem:condiiv}. This remark is important because 
Algorithm \ref{algo:stableBLU} spends some time in line 
\ref{line:clearblock} in clearing entries in order to make the computed 
matrix $L$ block unit lower triangular instead of simply unit lower 
triangular. In other words, commenting the line \ref{line:clearblock} in 
Algorithm \ref{algo:stableBLU} speeds up the execution of Algorithm 
\ref{algo:sheaf} but do not have any influence on its correctness.

\end{document}